\documentclass[onecolumn,12pt,draftclsnofoot]{IEEEtran}

\usepackage{mathrsfs}
\usepackage{mathptmx}
\usepackage{amsmath}
\usepackage{amsfonts}
\usepackage{amssymb}
\usepackage{graphicx}
\usepackage{xcolor}
\usepackage{float}
\usepackage{subfigure}
\usepackage{psfrag}
\usepackage{amsthm}
\usepackage{bm}

\newtheorem{thm}{Theorem}

\newtheorem{rem}{Remark}

\newtheorem{lem}{Lemma}
\def\Ds{\displaystyle}

\def\tp{\mathrm{T}}

\def\R{\mathbb{R}}
\def\Z{\mathbb{Z}}
\def\u{\mathbf{u}}
\def\1{\mathbf{1}}
\def\th{{\hat\theta}}

\begin{document}

\title{Optimization-Based Learning Control for Nonlinear Time-Varying Systems}
%
%

\author{Deyuan Meng and Jingyao Zhang
\thanks{The authors are with the Seventh Research Division, Beihang University (BUAA), Beijing 100191, P. R. China, and also with the School of Automation Science and Electrical Engineering, Beihang University (BUAA), Beijing 100191, P. R. China (e-mail: dymeng@buaa.edu.cn, zhangjingyao@buaa.edu.cn).}
}

\date{}
\maketitle

\begin{abstract}
Learning to perform perfect tracking tasks based on measurement data is desirable in the controller design of systems operating repetitively. This motivates the present paper to seek an optimization-based design approach for iterative learning control (ILC) of repetitive systems with unknown nonlinear time-varying dynamics. It is shown that perfect output tracking can be realized with updating inputs, where no explicit model knowledge but only measured input/output data are leveraged. In particular, adaptive updating strategies are proposed to obtain parameter estimations of nonlinearities. A double-dynamics analysis approach is applied to establish ILC convergence, together with boundedness of input, output, and estimated parameters, which benefits from employing properties of nonnegative matrices. Moreover, robust convergence is explored for optimization-based adaptive ILC in the presence of nonrepetitive uncertainties. Simulation tests are also implemented to verify the validity of our optimization-based adaptive ILC.
\end{abstract}


\section{Introduction}\label{sec1}

\IEEEPARstart{L}{earning} from measurement data but with no or limited model knowledge has become one of the most practically important problems in many application fields, such as robots, rail transportation, and batch processes. This motivates a class of learning control approaches designed by mainly resorting to the measurement data, rather than to the models for controlled plants. One of the most popular learning control approaches is proposed in \cite{akm:84} with a focus on acquiring the learning abilities of robots from repetitive executions (iterations, trials), leading to the so-called ``iterative learning control (ILC)'' that is simple and easy to implement even with limited plant knowledge. Due to the operation executed using only measurement data, ILC is considered as one of the natural data-driven control approaches \cite{hcg:17}. Since ILC is motivated from the physical learning patterns of human beings \cite{bta:06}, it is also catalogued as one of the typical intelligent control approaches \cite{acm:07}. In particular, ILC effectively applies to general nonlinear plants \cite{x:11}, and robustly works with the capability of rejecting the external disturbances, noises and initial shifts \cite{sw:14}.

One of the salient characteristics of ILC is to provide design tools to overcome shortcomings of conventional control design approaches. In particular, the design of ILC can be leveraged to improve the transient response performances for the controlled systems such that the perfect tracking objectives can be derived even in the presence of uncertain or unknown system structures and nonlinearities \cite{bta:06}-\cite{x:11}. This class of high performance tasks can be achieved over finite time steps gradually with increasing iterations. As a consequence, the convergence problem for ILC generally refers to the stability with respect to iteration because of the finite duration of time, which is considered as one of the key problems of ILC. There have been many effective methods to deal with ILC convergence problems, especially those based on the contraction mapping (CM) principle. To gain additional convergence properties, the optimization-based design together with CM-based analysis has been used as a good alternative for ILC (see, e.g., \cite{aor:96}-\cite{yxht:14}). It has been reported that optimization-based ILC can be designed to improve the convergence rate, or even accomplish the monotonic convergence to better transient learning behaviors.

In the literature, there have been introduced different classes of design approaches to optimization-based ILC. The first class is called norm-optimal approach that is developed by resorting to the lifted system representation of ILC (see, e.g., \cite{ba:11}-\cite{akn:14}). The norm-optimal ILC has wide potential applications for, e.g., robotic systems \cite{ba:11,ofc:14}, overhead cranes \cite{vds:13} and permanent magnet linear motors \cite{jps:13}, regarding which practical problems have also been discussed, such as robustness against repetitive model uncertainties \cite{sps:16,gse:18}, improvement of computational efficiencies \cite{hfv:12,sa:14} and extension to accommodate nonlinear dynamics \cite{vds:13,akn:14}. The second class is devoted to stochastic ILC such that the optimization-based design can be explored to overcome ill effects arising from random (iteration-dependent) disturbances and noises (see, e.g., \cite{s:01a}-\cite{rv:10}). It is worth noting that all aforementioned optimization-based ILC approaches are either focused directly on linear systems \cite{ba:11,ofc:14,jps:13}-\cite{sa:14,s:01a}-\cite{rv:10} or extended from linear systems to nonlinear systems with known linearized models \cite{vds:13,akn:14}. By contrast, the third class of optimization-based ILC approaches has been exploited by directly dealing with nonlinear systems subject to unknown nonlinearities, which creates data-driven or model free optimal ILC (see, e.g., \cite{ch:07}-\cite{hchh:19}).

The data-driven optimal ILC requires no explicit models for algorithms design and convergence analyses, which is achieved by combining a dynamical linearization approach for nonlinear systems with an adaptive estimation approach for linearization parameters \cite{ch:07}-\cite{hchh:19}. This also leads to a type of optimization-based adaptive ILC that permits not only the nonlinear systems but also their dynamical linearization models to have unknown dynamics and model structures. Furthermore, the optimization-based adaptive ILC has a property that its convergence analysis can be developed through the CM approach, especially through the eigenvalue-based CM approach. The eigenvalue analysis is well known as an easy-to-implement and popular approach for ILC convergence. However, despite these good properties, the eigenvalue-based CM approach is restricted to ILC processes with iteration-independent parameters based on the basic linear system theory \cite{r:96,mjd:15}.

It is worth emphasizing that for nonlinear control plants, the dynamical linearization inevitably leads to iteration-dependent model parameters \cite{ch:07}-\cite{hchh:19}. This renders the eigenvalue-based CM approach no longer effective in implementing convergence analysis of optimization-based adaptive ILC. Another issue left to settle for optimization-based adaptive ILC is robustness with regard to iteration-dependent uncertainties that is considered to be practically important for ILC \cite{mm:17}-\cite{yhh:16}. Actually, the robust issue has not been well studied for optimization-based adaptive ILC (see, e.g., \cite{ch:07}-\cite{byhq:19}). It is mainly due to that the iteration-dependent uncertainties may bring challenging difficulties into ILC convergence in the presence of nonrepetitiveness created by iteration-dependent model parameters. To accommodate the effects arising from nonrepetitiveness, new design and analysis approaches for ILC usually need to be explored, see, e.g., \cite{hchh:19} for an extended state observer-based design approach and \cite{mm:17,mz:19} for a double-dynamics analysis (DDA) approach. Despite these new approaches, the eigenvalue analysis is still leveraged in \cite{hchh:19}, and linear systems are only addressed in \cite{mm:17,mz:19}.

In this paper, we contribute to exploiting optimization-based ILC for nonlinear systems, in which we particularly propose an adaptive updating law for estimation of unknown time-varying nonlinearities. It is shown that the boundedness of all estimated parameters can be ensured directly form an optimization-based design. Further, the ILC convergence is achieved, together with the boundedness of system trajectories, for which we introduce a DDA approach by leveraging good properties of nonnegative matrices. We also explore the robustness of optimization-based adaptive ILC with respect to nonrepetitive uncertainties caused from iteration-dependent disturbances and initial shifts. Based on comparisons with the relevant existing results, the following main contributions are summarized for our optimization-based adaptive ILC.
\begin{enumerate}
\item
We propose a new design method for optimization-based adaptive ILC of nonlinear time-varying systems. It yields a data-driven optimal ILC algorithm that however differs from those of, e.g., \cite{ch:07}-\cite{hchh:19}, especially for the updating law of parameter estimation. Additionally, an advantage of the new design method is that all estimated parameters are naturally ensured to be bounded.

\item
We propose a new analysis method to settle convergence problems of optimization-based adaptive ILC. It benefits from implementing a DDA-based approach to ILC based on properties of nonnegative matrices. A consequence of this is the exploration of selection conditions for learning parameters such that we not only exploit the boundedess of system trajectories but also achieve the perfect output tracking tasks. Furthermore, our ILC convergence results avoid performing the eigenvalue analysis that is required in, e.g., \cite{ch:07}-\cite{chhj:18,hchh:19}.

\item
We develop robust convergence analysis of optimization-based adaptive ILC for nonlinear systems in the presence of nonrepetitive uncertainties. It is shown that our design and analysis methods can be generalized to overcome the effect arising from iteration-dependent disturbances and initial shifts. In comparison with this, the robust problem has not been well addressed in, e.g., \cite{ch:07}-\cite{byhq:19}.
\end{enumerate}

In addition, we carry out simulation tests to demonstrate the effectiveness of our algorithm that optimization-based adaptive ILC both guarantees the boundedness of all system trajectories and achieves the prescribed perfect tracking tasks. Further, the robust performances are also illustrated for our optimization-based adaptive ILC, regardless of disturbances and initial shifts that are varying with respect to both iteration and time.

We organize the remainder sections of this paper as follows. In Section \ref{sec2}, we present the optimization-based ILC problem, for which an algorithm of optimization-based adaptive ILC is designed in Section \ref{sec3}. The main ILC convergence results are established in Section \ref{sec4}, and are further generalized to carry out robust analysis with respect to nonrepetitive uncertainties in Section \ref{sec5}. Simulations are performed, and then conclusions are made, in Sections \ref{sec6} and \ref{sec7}, respectively. The proofs of all lemmas are given in appendices.

{\it Notations:} Let $\Z_{+}=\{0,1,2,\cdots\}$, $\Z=\{1,2,3,\cdots\}$, $\Z_{T}=\{0,1,\cdots,T\}$ with any $T\in\Z$, and $\1_n=\left[1,1,\cdots,1\right]^{\tp}\in\mathbb{R}^{n}$. For a matrix $A=\left[a_{ij}\right]\in\mathbb{R}^{n\times m}$, $\|A\|$ denotes any norm of $A$, where specifically $\|A\|_{\infty}$ and $\|A\|_{2}$ are the maximum row sum matrix norm and the spectral norm of $A$, respectively. Let $m=1$, and then $\|A\|_{\infty}$ and $\|A\|_{2}$ become the $l_{\infty}$ and $l_{2}$ norms of a vector $A\in\mathbb{R}^{n}$, respectively. When $m=n$, $\rho(A)$ denotes the spectral radius of a square matrix $A\in\mathbb{R}^{n\times n}$. We call $A$ a nonnegative matrix if $a_{ij}\geq0$, $\forall i=1$, $2$, $\cdots$, $n$, $\forall j=1$, $2$, $\cdots$, $m$, which is denoted by $A\geq0$. A trivial nonnegative matrix induced by $A$ is $|A|=\left[\left|a_{ij}\right|\right]\geq0$, and for any two matrices $A$, $B\in\mathbb{R}^{n\times m}$, $A\geq B$ means $A-B\geq0$. For a matrix sequence $\{A_{i}\in\mathbb{R}^{n\times m}:i\in\Z_{+}\}$, let $\sum_{i=h}^{j}A_{i}=0$ (i.e., the null matrix of appropriate dimensions) if $j<h$, and for $m=n$, let $\prod_{i=h}^{j}A_{i}=A_{j}A_{j-1}\cdots A_{h}$ if $j\geq h$ and $\prod_{i=h}^{j}A_{i}=I$ (i.e., the identity matrix of appropriate dimensions) if $j<h$. A difference operator of a vector $\tau_{k}(t)\in\mathbb{R}^{n}$ is defined as $\Delta:\tau_{k}(t)\to\Delta\tau_{k}(t)=\tau_{k}(t)-\tau_{k-1}(t)$, $\forall k\in\Z$, $\forall t\in\Z_{+}$.

\section{Problem Statement}\label{sec2}

Consider a class of nonlinear discrete-time-varying systems with input-output dynamics described by
\begin{equation}\label{eq01}
\aligned
y_{k}(t+1)&=f\left(y_{k}(t),\cdots,y_{k}(t-l),u_{k}(t),\cdots,u_{k}(t-n),t\right)\\
\hbox{with}~
&y_{k}(i)
=\left\{\aligned
&0,~~i<0\\
&y_{0},~i=0
\endaligned\right.~\hbox{and}~
u_k(i)=0,~i<0
\endaligned
\end{equation}

\noindent where $t\in\Z_{T-1}$ and $k\in\Z_+$ are the time and iteration indexes, respectively; $y_{k}(t)\in\R$ and $u_{k}(t)\in\R$ are the output and input, respectively; $l\in\Z_+$ and $n\in\Z_+$ are nonnegative integers that represent the system output and input orders, respectively; and \[f:\underbrace{\R\times\R\times\cdots\times\R}_{l+n+3}\to\R\]

\noindent is an unknown nonlinear function. For the sake of convenience, we write this nonlinear function as $f$ or $f\left(x_{1},x_{2},\cdots,x_{l+n+3}\right)$, where $x_{i}\in\R$, $i=1,2,\cdots,l+n+3$ denotes the $i$th independent variable of $f$.

{\bf Problem Statement.} Given any desired reference trajectory $y_d(t)\in\R$ over $t\in\Z_{T}$, the objective of this paper is to design an ILC algorithm based on solving an optimization problem such that the uncertain nonlinear system (\ref{eq01}) achieves the following perfect tracking task:
\begin{equation}\label{eq07}
\aligned
\lim_{k \to \infty} y_k(t)=y_d(t),~~~ \forall t\in\Z_{T}\setminus\{0\}.
\endaligned
\end{equation}

\noindent Correspondingly, we are interested in the optimization problem by leveraging the following index over $t\in\Z_{T-1}$ and $k\in\Z$ (see also \cite{chhj:15,chjh:18}):
\begin{equation}\label{eq08}
J\left(u_k(t)\right)=\left[\sum_{i=1}^{m}\gamma_i e_{k-i+1}(t+1)\right]^{2}+\lambda\left[\Delta u_k(t)\right]^{2}
\end{equation}

\noindent where $e_{k}(t)=y_{d}(t)-y_{k}(t)$ denotes the (output) tracking error, $\Delta u_{k}(t)=u_{k}(t)-u_{k-1}(t)$ represents the input error between two sequential iterations, and $\lambda>0$ and $\gamma_{i}>0$, $i=1$, $2$, $\cdots$, $m$ are some positive learning parameters. In (\ref{eq08}), we consider a high order $m\in\Z$ for the tracking errors of interest over iterations, and adopt $e_{i}(t+1)=0$, $\forall t\in\Z_{T-1}$ if $i<0$.

To address the abovementioned ILC problem, we introduce a fundamental assumption for the continuous differentiability of the unknown nonlinear function $f$.
\begin{enumerate}
\item [(A1)]
Let $f$ be continuously differentiable such that the partial derivatives with respect to the first $l+n+2$ independent variables are bounded, namely,
\begin{equation}\label{eq02}
\aligned
&\left|\frac{\partial f}{\partial x_{i}}\left(x_{1},x_{2},\cdots,x_{l+n+2},t\right)\right|
\leq\beta_{\overline{f}},
\quad\forall x_{i}\in\R,i=1,2,\cdots,l+n+2,~~\forall t\in\Z_{T-1}
\endaligned
\end{equation}

\noindent where $\beta_{\overline{f}}>0$ is some finite bound. Further, let the input-output coupling function, defined by ${\Ds\partial f}/{\Ds\partial x_{l+2}}$, be sign-fixed, which without any loss of generality is considered to be positive, namely,
\begin{equation}\label{eq03}
\aligned
&\frac{\partial f}{\partial x_{l+2}}\left(x_{1},x_{2},\cdots,x_{l+n+2},t\right)\geq\beta_{\underline{f}},
\quad\forall x_{i}\in\R,i=1,2,\cdots,l+n+2,~~\forall t\in\Z_{T-1}
\endaligned
\end{equation}

\noindent for some finite bound $\beta_{\underline{f}}>0$.
\end{enumerate}

\begin{rem}\label{rem1}
In general, the globally Lipschitz condition is one of the basic requirements of nonlinear ILC \cite{hcg:17}-\cite{x:11}, which can be satisfied for the nonlinear system (\ref{eq01}) under the Assumption (A1). To be specific, if we apply the mean value theorem (see, e.g., \cite[P. 651]{k:02}), then for any $x_{i}$ and $\bar{x}_{i}$, $i=1$, $2$, $\cdots$, $l+n+2$, there exists some $z_{i}=\sigma x_{i}+(1-\sigma)\bar{x}_{i}$ with $\sigma\in[0,1]$ such that
\begin{equation*}
\aligned
&f\left(x_{1},x_{2},\cdots,x_{l+n+2},t\right)
-f\left(\bar{x}_{1},\bar{x}_{2},\cdots,\bar{x}_{l+n+2},t\right)
=\sum_{i=1}^{l+n+2}\frac{\Ds\partial f}{\Ds\partial x_{i}}\bigg|_{\left(z_{1},z_{2},\cdots,z_{l+n+2},t\right)}\left(x_{i}-\bar{x}_{i}\right),~~~\forall t\in\Z_{T-1}
\endaligned
\end{equation*}

\noindent which, together with (\ref{eq02}), leads to
\begin{equation*}
\aligned
&\left|f\left(x_{1},x_{2},\cdots,x_{l+n+2},t\right)
-f\left(\bar{x}_{1},\bar{x}_{2},\cdots,\bar{x}_{l+n+2},t\right)\right|
\leq\beta_{\overline{f}}\sum_{i=1}^{l+n+2}\left|x_{i}-\bar{x}_{i}\right|,~~~\forall t\in\Z_{T-1}.
\endaligned
\end{equation*}

\noindent Since less plant information on the uncertain nonlinear system (\ref{eq01}) is known, an adaptive ILC law is generally needed to reach the perfect tracking objective (\ref{eq07}) via handling the optimization problem with the index (\ref{eq08}) (see, e.g., \cite{ch:07}-\cite{chjh:18}). This requires the sign of the system input-output coupling function ${\Ds\partial f}/{\Ds\partial x_{l+2}}$ to be fixed in the optimization-based design of adaptive ILC, as made in the Assumption (A1). Particularly, we can see from (\ref{eq02}) and (\ref{eq03}) that ${\Ds\partial f}/{\Ds\partial x_{l+2}}$ is not only sign-fixed but also bounded. Namely, we have
\begin{equation}\label{eq06}
\aligned
\frac{\Ds\partial f}{\Ds\partial x_{l+2}}&\left(x_{1},x_{2},\cdots,x_{l+n+2},t\right)\in\left[\beta_{\underline{f}},\beta_{\overline{f}}\right],
\quad\forall x_{i}\in\R,i=1,2,\cdots,l+n+2,\quad\forall t\in\Z_{T-1}.
\endaligned
\end{equation}
\end{rem}

\section{Optimization-Based Adaptive ILC}\label{sec3}

In this section, we present a design method for optimization-based adaptive ILC, regardless of controlled systems subject to unknown nonlinear time-varying dynamics. We thus propose a helpful lemma to develop an extended dynamical linearization for the unknown nonlinear time-varying dynamics such that we may realize an adaptive ILC design by solving the optimization problem with the index (\ref{eq08}).

\begin{lem}\label{lem1}
For the nonlinear system (\ref{eq01}) under the Assumption (A1), an extended dynamical linearization can be given by
\begin{equation}\label{eq09}
\aligned
\begin{bmatrix}
  y_{i}(1) \\
  y_{i}(2) \\
  \vdots \\
  y_{i}(T)
\end{bmatrix}-\begin{bmatrix}
  y_{j}(1) \\
  y_{j}(2) \\
  \vdots \\
  y_{j}(T)
\end{bmatrix}
=\Theta_{i,j}\left(\begin{bmatrix}
  u_{i}(0) \\
  u_{i}(1) \\
  \vdots \\
  u_{i}(T-1)
\end{bmatrix}
-\begin{bmatrix}
  u_{j}(0) \\
  u_{j}(1) \\
  \vdots \\
  u_{j}(T-1)
\end{bmatrix}\right)
\endaligned
\end{equation}

\noindent together with $\Theta_{i,j}$, $\forall i$, $j\in\Z_{+}$ being given in a lower triangular matrix form of
\begin{equation*}
\aligned
\Theta_{i,j}=\begin{bmatrix}
                \theta_{i,j,0}(0) & 0 &\cdots & 0 \\
                \theta_{i,j,1}(0) & \theta_{i,j,1}(1) & \ddots & \vdots \\
                \vdots & \ddots & \ddots & 0\\
                \theta_{i,j,T-1}(0)  &\cdots & \cdots & \theta_{i,j,T-1}(T-1)
              \end{bmatrix}
\endaligned
\end{equation*}

\noindent of which all nonzero entries can be guaranteed to be bounded, namely, for some finite bound $\beta_{\theta}>0$,
\begin{equation}\label{eq8}
\left|\theta_{i,j,t}(\xi)\right|
\leq\beta_{\theta},\quad\forall\xi\in\Z_{t},\forall t\in\Z_{T-1},\forall i,j\in\Z_{+}.
\end{equation}

\noindent In particular, the diagonal entries of $\Theta_{i,j}$ satisfy
\begin{equation}\label{eq04}
\theta_{i,j,t}(t)
\in\left[\beta_{\underline{f}},\beta_{\overline{f}}\right],\quad\forall t\in\Z_{T-1},\forall i,j\in\Z_{+}.
\end{equation}
\end{lem}

\begin{proof}
This lemma can be established based on the differential mean value theorem and with the derivation rules of compound functions, where the facts of (\ref{eq02}) and (\ref{eq06}) should be noticed. For the proof details, we refer the readers to Appendix \ref{apdx1}.
\end{proof}

Note that to derive (\ref{eq8}), we generally have $\beta_{\theta}\geq\beta_{\overline{f}}$ in Lemma \ref{lem1}. This leads to the estimation of $\theta_{i,j,t}(t)$ in a more reasonable form (\ref{eq04}), rather than $\theta_{i,j,t}(t)\in\left[\beta_{\underline{f}},\beta_{\theta}\right]$, $\forall t\in\Z_{T-1}$, $\forall i$, $j\in\Z_{+}$. As an application of Lemma \ref{lem1}, we focus upon the input-output relationship between two sequential iterations $k$ and $k-1$ for the nonlinear system (\ref{eq01}). Namely, by letting $i=k$ and $j=k-1$ in (\ref{eq09}), we can obtain
\begin{equation}\label{eq10}
\aligned
\Delta y_{k}(t+1)
&=\sum_{i=0}^{t}\theta_{k,k-1,t}(i)\Delta u_{k}(i)\\
&\triangleq\Delta\overrightarrow{u_{k}}^{\tp}(t)\overrightarrow{\theta_{k,k-1,t}}(t),\quad\forall t\in\Z_{T-1},\forall k\in\Z
\endaligned
\end{equation}

\noindent where $\overrightarrow{u_{k}}(t)$ and $\overrightarrow{\theta_{k,k-1,t}}(t)$ are defined, respectively, from $u_{k}(t)$ and $\theta_{k,k-1,t}(t)$ as
\[\aligned
\overrightarrow{u_{k}}(t)&=\left[u_{k}(0),u_{k}(1),\cdots,u_{k}(t)\right]^{\tp}\\
\overrightarrow{\theta_{k,k-1,t}}(t)&=\left[\theta_{k,k-1,t}(0),\theta_{k,k-1,t}(1),\cdots,\theta_{k,k-1,t}(t)\right]^{\tp}.
\endaligned
\]

\noindent In view of this discussion, we proceed to explore the ``optimal solution'' for the nonlinear system (\ref{eq01}) to solve the optimization problem with the index (\ref{eq08}).

\begin{lem}\label{lem7}
For the nonlinear system (\ref{eq01}) under the Assumption (A1), the solution that optimizes the index (\ref{eq08}) for $t\in\Z_{T-1}$ and $k\in\Z$ can be presented in an updating form of
\begin{equation}\label{eq14}
\aligned
u_{k}(t)
&=u_{k-1}(t)
-\frac{\Ds\gamma_{1}^{2}\theta_{k,k-1,t}(t)}{\Ds\lambda+\gamma_{1}^{2}\theta_{k,k-1,t}^{2}(t)}
\sum_{i=0}^{t-1}\theta_{k,k-1,t}(i)\big[u_{k}(i)-u_{k-1}(i)\big]\\
&~~~+\frac{\Ds\gamma_{1}\theta_{k,k-1,t}(t)}{\Ds\lambda+\gamma_{1}^{2}\theta_{k,k-1,t}^{2}(t)}
\Bigg[\gamma_{1}e_{k-1}(t+1)
+\sum_{i=2}^{m}\gamma_{i}e_{k-i+1}(t+1)\Bigg].
\endaligned
\end{equation}
%
\end{lem}

\begin{proof}
To optimize the index (\ref{eq08}) for the nonlinear system (\ref{eq01}), we resort to the condition ${\Ds\partial J(u_k(t))}/{\Ds\partial u_k(t)}=0$ and then based on leveraging the fact of (\ref{eq10}), we can deduce (\ref{eq14}). The proof details are given in Appendix \ref{apdx3}.
\end{proof}

Even though the input $u_{k}(t)$ determined by (\ref{eq14}) is capable of optimizing the index (\ref{eq08}), the implementation of (\ref{eq14}) resorts to the exact information of $\theta_{k,k-1,t}(i)$, $\forall k\in\Z$, $\forall t\in\Z_{T-1}$, $\forall i\in\Z_{t}$. However, these parameters are unavailable, which are induced by the dynamical linearization of the unknown nonlinear time-varying dynamics involved in (\ref{eq01}) (see Lemma \ref{lem1}). To overcome this problem, we present a parameter estimation algorithm with respect to iteration such that we can develop an estimated value $\th_{k,k-1,t}(i)$ of $\theta_{k,k-1,t}(i)$, $\forall k\in\Z$, $\forall t\in\Z_{T-1}$, $\forall i\in\Z_{t}$. We explore an optimization-based approach to calculating these estimation parameters based on the following index:
\begin{equation}\label{eq011}
\aligned
H\left(\overrightarrow{\hat{\theta}_{k,k-1,t}}(t)\right)
&=\left[\Delta y_{k-1}(t+1)-\Delta\overrightarrow{u_{k-1}}^{\tp}(t)\overrightarrow{\hat{\theta}_{k,k-1,t}}(t)\right]^{2}\\
&~~~+\mu_{1}\left\|\overrightarrow{\hat{\theta}_{k,k-1,t}}(t)-\overrightarrow{\hat{\theta}_{k-1,k-2,t}}(t)\right\|_{2}^{2}
+\mu_{2}\left\|\overrightarrow{\hat{\theta}_{k,k-1,t}}(t)\right\|_{2}^{2},\quad\forall t\in\Z_{T-1},\forall k\geq2
\endaligned
\end{equation}

\noindent where $\overrightarrow{\hat{\theta}_{k,k-1,t}}(t)$ is the estimation of $\overrightarrow{\theta_{k,k-1,t}}(t)$, defined as
\[\overrightarrow{\hat{\theta}_{k,k-1,t}}(t)=\left[\hat{\theta}_{k,k-1,t}(0),\hat{\theta}_{k,k-1,t}(1),\cdots,\hat{\theta}_{k,k-1,t}(t)\right]^{\tp}
\]

\noindent and $\mu_{1}>0$ and $\mu_{2}>0$ denote two positive weighting factors. It is worth emphasizing that (\ref{eq011}) shows a new optimization index to accomplish the estimation of unknown system parameters in optimization-based adaptive ILC, by contrast with the existing relevant results of, e.g., \cite{ch:07}-\cite{hchh:19}.

\begin{rem}\label{rem2}
For three terms involved in (\ref{eq011}), the first term is to provide $\overrightarrow{\theta_{k,k-1,t}}(t)$ with a reasonable estimation $\overrightarrow{\hat{\theta}_{k,k-1,t}}(t)$, the second term is to render $\overrightarrow{\hat{\theta}_{k,k-1,t}}(t)$ a slowly varying estimation along the iteration axis, while the third term is to guarantee the boundedness of $\overrightarrow{\hat{\theta}_{k,k-1,t}}(t)$. It is further expected that $\overrightarrow{\hat{\theta}_{k,k-1,t}}(t)$ does not converge to zero with the increasing of iterations. This can be realized through the selections of two weighting factors in (\ref{eq011}). In fact, the bigger each of the weighting factors is, the more important its weighted term is in (\ref{eq011}). As the candidates, $\mu_{1}=1$ may be directly employed without any loss of generality, whereas a relatively small $\mu_{2}$ should be adopted to both ensure $\overrightarrow{\hat{\theta}_{k,k-1,t}}(t)$ to be bounded and avoid it to converge to zero (e.g., $\mu_{2}=0.001$ may be used).
\end{rem}

To proceed further with the index (\ref{eq011}) by finding its optimal solution, the following lemma presents an updating law for the parameter estimation.

\begin{lem}\label{lem8}
For $t\in\Z_{T-1}$ and $k\geq2$, the solution that optimizes the index (\ref{eq011}) can be proposed in an updating form of
\begin{equation}\label{eq012}
\aligned
\overrightarrow{\hat{\theta}_{k,k-1,t}}(t)
&=\frac{\Ds\mu_{1}}{\Ds\mu_{1}+\mu_{2}}\overrightarrow{\hat{\theta}_{k-1,k-2,t}}(t)\\
&~~~+\frac{\Ds1}{\Ds\mu_{1}+\mu_{2}+\left\|\Delta\overrightarrow{{u}_{k-1}}(t)\right\|_{2}^{2}}\bigg[\Delta y_{k-1}(t+1)
-\frac{\Ds\mu_1}{\Ds\mu_{1}+\mu_{2}}\Delta\overrightarrow{{u}_{k-1}}^{\tp}(t)\overrightarrow{\hat{\theta}_{k-1,k-2,t}}(t)\bigg]
\Delta\overrightarrow{{u}_{k-1}}(t).
\endaligned
\end{equation}
\end{lem}

\begin{proof}
This lemma can be established by optimizing the index (\ref{eq011}) via a similar idea as the proof of Lemma \ref{lem7}. See Appendix \ref{apdx4} for the proof details.
\end{proof}

We are in position to leverage the development of Lemmas \ref{lem7} and \ref{lem8} to propose an optimization-based adaptive ILC algorithm for the uncertain nonlinear system (\ref{eq01}).
\[
\begin{array}{l}
\hbox{\hspace{19pt}Algorithm 1: Optimization-Based Adaptive ILC}\vspace{8pt}\\\hline\vspace{-12pt}\\
\hbox{1) do step (S1);}\vspace{8pt}\\\hline\vspace{-12pt}\\
\hbox{2) let $k=1$, and go to step 3) to start iteration;}\vspace{8pt}\\\hline\vspace{-12pt}\\
\hbox{3) apply $u_{k-1}(t)$ to operate the nonlinear system (1);}\vspace{8pt}\\\hline\vspace{-12pt}\\
\hbox{4) do step (S2) if $k\geq2$; otherwise, go directly to step 5);}\vspace{8pt}\\\hline\vspace{-12pt}\\
\hbox{5) do step (S3);}\vspace{3pt}\\\hline\vspace{-12pt}\\
\hbox{6) let $k=k+1$, and go back to step 3);}\vspace{8pt}\\\hline\vspace{-10pt}\\
\end{array}
\]

\noindent in which the steps (S1), (S2) and (S3) are presented as follows.
\begin{enumerate}
\item[(S1)]
For any $t\in\Z_{T-1}$, choose any bounded initial input $u_{0}(t)$ and initial estimated value $\th_{1,0,t}(i)$ of $\theta_{1,0,t}(i)$, $\forall i\in\Z_{t}$. In particular, given any (small) scalar $\epsilon>0$, choose $\th_{1,0,t}(t)$ such that
\begin{equation}\label{eq16}
\aligned
\th_{1,0,t}(t)\geq\epsilon,\quad\forall t\in\Z_{T-1}.
\endaligned
\end{equation}

\item[(S2)]
For any $t\in\Z_{T-1}$, apply an updating law of the parameter estimation with respect to each iteration $k\geq2$ and each time step $i\in\Z_{t}$ as
\begin{equation}\label{eq15}
\aligned
\th_{k,k-1,t}(i)
&=\frac{\Ds\mu_{1}}{\Ds\mu_{1}+\mu_{2}}\th_{k-1,k-2,t}(i)\\
&~~~+\frac{\Ds\Delta u_{k-1}(i)}{\Ds\mu_{1}+\mu_{2}+\sum_{j=0}^{t}\Delta u_{k-1}^2(j)}\bigg[\Delta y_{k-1}(t+1)
-\frac{\Ds\mu_{1}}{\Ds\mu_{1}+\mu_{2}}\sum_{j=0}^{t}\th_{k-1,k-2,t}(j)\Delta u_{k-1}(j)\bigg].
\endaligned
\end{equation}

\noindent In particular, if $\th_{k,k-1,t}(t)<\epsilon$, then set
\begin{equation}\label{eq010}
\th_{k,k-1,t}(t)=\th_{1,0,t}(t),\quad\forall t\in\Z_{T-1}.
\end{equation}

\item[(S3)]
For any $t\in\Z_{T-1}$, apply an updating law with respect to the input for each iteration $k\in\Z$ as
\begin{equation}\label{eq17}
\aligned
u_{k}(t)
&=u_{k-1}(t)
-\frac{\Ds\gamma_{1}^{2}\th_{k,k-1,t}(t)}{\Ds\lambda+\gamma_{1}^{2}\th^{2}_{k,k-1,t}(t)}
\sum_{i=0}^{t-1}\th_{k,k-1,t}(i)\big[u_{k}(i)-u_{k-1}(i)\big]\\
&~~~+\frac{\Ds\gamma_{1}\th_{k,k-1,t}(t)}{\Ds\lambda+\gamma_{1}^{2}\th^{2}_{k,k-1,t}(t)}
\Bigg[\gamma_{1}e_{k-1}(t+1)
+\sum_{i=2}^{m}\gamma_{i}e_{k-i+1}(t+1)\Bigg].
\endaligned
\end{equation}
\end{enumerate}

\begin{rem}\label{rem3}
From (\ref{eq16}) and (\ref{eq010}), we can obtain $\th_{k,k-1,t}(t)\geq\epsilon$, $\forall k\in\Z$, $\forall t\in\Z_{T-1}$. This discloses that $\epsilon$ represents the smallest acceptable value of the estimation $\th_{k,k-1,t}(t)$ for the parameter $\theta_{k,k-1,t}(t)$, $\forall k\in\Z$, $\forall t\in\Z_{T-1}$. The application of Algorithm 1 thus naturally avoids the zero convergence of $\th_{k,k-1,t}(t)$ along the iteration axis. In particular, when noting (\ref{eq03}) or (\ref{eq06}), we can also find that (\ref{eq16}) guarantees $\th_{1,0,t}(t)$ to have the same sign as ${\Ds\partial f}/{\Ds\partial x_{l+2}}$. It is important for the practical implementation of adaptive ILC. In addition, the Algorithm 1 greatly generalizes the existing optimization-based adaptive ILC algorithms in the literature. For example, if we set $\mu_{2}=0$, then the Algorithm 1 collapses into the high-order adaptive ILC algorithm of, e.g., \cite{chhj:15,chjh:18}; and furthermore, if we take $m=1$, then it becomes the first-order adaptive ILC algorithm of, e.g., \cite{ch:07,hj:13}.
\end{rem}
\begin{rem}\label{rem4}
It is worth highlighting that all the parameters $\lambda$, $\gamma_{i}$, $i=1,2,\cdots,m$, $\mu_{1}$ and $\mu_{2}$ involved in the Algorithm 1 can be considered to be time-varying. This is a trivial generalization of our optimization-based adaptive ILC design and will neither affect our following convergence analysis. To get the updating law (\ref{eq15}), we introduce a novel optimization-based approach to the estimations of parameters in Lemma \ref{lem8}. As a direct benefit, (\ref{eq15}) avoids adding an additional step-size factor to guarantee the boundedness of estimated parameters in comparison with the existing results of, e.g., \cite{chhj:15,chjh:18}.
\end{rem}

\section{Convergence Analysis Results}\label{sec4}

We next contribute to exploring the convergence analysis of the nonlinear system (\ref{eq01}) that operates under the Algorithm 1 of optimization-based adaptive ILC. Toward this end, we resort to the tracking error and can employ (\ref{eq10}) to equivalently derive
\begin{equation}\label{eq11}
\aligned
e_{k}(t+1)
&=e_{k-1}(t+1)-\Delta y_{k}(t+1)\\
&=e_{k-1}(t+1)
-\sum_{i=0}^{t}\theta_{k,k-1,t}(i)\Delta u_k(i),\quad\forall t\in\Z_{T-1},\forall k\in\Z
\endaligned
\end{equation}

\noindent in which nonrepetitive (namely, iteration-dependent) uncertain parameters $\theta_{k,k-1,t}(i)$, $\forall i\in\Z_{t}$, $\forall t\in\Z_{T-1}$, $\forall k\in\Z$ are inevitably involved. It may result in challenging difficulties for exploiting robust convergence results of ILC. For example, the eigenvalue (or spectral radius) analysis is not applicable any longer when the system (matrix) parameters of the resulting ILC process are explicitly dependent upon iteration (see \cite{mjd:15} for more detailed discussions). The traditional CM-based method of convergence analysis may even be not effective in ILC due to nonrepetitive uncertainties (see also \cite{mm:17,mz:19}).

To make the abovementioned observations clearer to follow, we insert (\ref{eq17}) into (\ref{eq11}) and can further deduce
\begin{equation}\label{eq25}
\aligned
e_{k}(t+1)
&=e_{k-1}(t+1)-\theta_{k,k-1,t}(t)\Delta u_{k}(t)
-\sum_{i=0}^{t-1}\theta_{k,k-1,t}(i)\Delta u_k(i)\\
&=\left[1-\frac{\Ds\left(\gamma_{1}^{2}+\gamma_{1}\gamma_{2}\right)\theta_{k,k-1,t}(t)\th_{k,k-1,t}(t)}
{\Ds\lambda+\gamma_{1}^{2}\th^{2}_{k,k-1,t}(t)}\right]e_{k-1}(t+1)\\
&~~~-\sum_{i=3}^{m}\frac{\Ds\gamma_{1}\gamma_{i}\theta_{k,k-1,t}(t)\th_{k,k-1,t}(t)}{\Ds\lambda+\gamma_{1}^{2}\th^{2}_{k,k-1,t}(t)}e_{k-i+1}(t+1)
+\kappa_{k}(t)
\endaligned
\end{equation}

\noindent where $\kappa_k(t)$ is a driving signal given by
\begin{equation}\label{eq26}
\aligned
\kappa_{k}(t)
&=\frac{\Ds\gamma_{1}^{2}\theta_{k,k-1,t}(t)\th_{k,k-1,t}(t)}{\Ds\lambda+\gamma_{1}^{2}\th^{2}_{k,k-1,t}(t)}
\sum_{i=0}^{t-1}\th_{k,k-1,t}(i)\Delta u_{k}(i)
-\sum_{i=0}^{t-1}\theta_{k,k-1,t}(i)\Delta u_{k}(i).
\endaligned
\end{equation}

\noindent Obviously, we can see from (\ref{eq25}) that the system parameters of the ILC process resulting from the nonlinear system (\ref{eq01}) under the Algorithm 1 depend explicitly on $\theta_{k,k-1,t}(t)$ and $\th_{k,k-1,t}(t)$ and, hence, are iteration-dependent. This renders the traditional CM-based method not applicable to ILC convergence analysis, especially those CM-based methods using eigenvalue analyses, to overcome which we apply a DDA approach to optimization-based adaptive ILC by leveraging the properties of nonnegative matrices (see \cite[Chapter 8]{hj:85}).

\subsection{Boundedness of Estimated System Parameters}

As noted in (\ref{eq25}) and (\ref{eq26}), the uncertain parameter $\theta_{k,k-1,t}(i)$ and its estimation $\th_{k,k-1,t}(i)$, $\forall i\in\Z_{t}$ both play crucial roles in optimization-based adaptive ILC along the iteration axis. With Lemma \ref{lem1}, we can obtain a basic boundedness property of each parameter $\theta_{k,k-1,t}(i)$, $\forall i\in\Z_{t}$, whereas we gain the estimation $\th_{k,k-1,t}(i)$, $\forall i\in\Z_{t}$ in the process of applying the Algorithm 1 to the nonlinear system (\ref{eq01}), for which it is needed to determine whether the basic boundedness property holds. An affirmative answer to this question is provided in the following theorem.

\begin{thm}\label{thm2}
For the nonlinear system (\ref{eq01}), let the Assumption (A1) hold. If the Algorithm 1 is applied, then the boundedness of the estimation $\th_{k,k-1,t}(i)$ can be guaranteed such that 
\begin{equation}\label{eq18}
\left|\th_{k,k-1,t}(i)\right|\leq\beta_{\th},\quad\forall i\in\Z_{t},\forall t\in\Z_{T-1},\forall k\in\Z
\end{equation}

\noindent for some finite bound $\beta_{\th} >0$. In particular, $\th_{k,k-1,t}(t)$ satisfies 
\begin{equation}\label{eq018}
\th_{k,k-1,t}(t)\in\left[\epsilon,\beta_{\th}\right],\quad\forall t\in\Z_{T-1},\forall k\in\Z.
\end{equation}
\end{thm}

With Theorem \ref{thm2}, it is revealed that the estimated parameters of all the nonrepetitive uncertain parameters $\theta_{k,k-1,t}(i)$, $\forall i\in\Z_{t}$, $\forall t\in\Z_{T-1}$, $\forall k\in\Z$ are bounded when employing the Algorithm 1 for the nonlinear system (\ref{eq01}). This boundedness result resorts to no conditions on the input updating law (\ref{eq17}), which is even independent of the selections of the learning parameters $\lambda$ and $\gamma_{i}$, $i=1$, $2$, $\cdots$, $m$. Furthermore, Theorem \ref{thm2} is naturally ensured with $\mu_{1}>0$ and $\mu_{2}>0$ in the updating law (\ref{eq15}) for parameter estimation but without adding a step-size factor in (\ref{eq15}), which is different from, e.g., \cite{chhj:15,chjh:18}.

To prove Theorem \ref{thm2}, a useful lemma on the norm estimation of an iteration-dependent matrix operator is given as follows.

\begin{lem}\label{lem2}
For any $t\geq0$ and $k\geq2$, $\left\|Q\left(\Delta\overrightarrow{u_{k-1}}(t)\right)\right\|_{2}\leq1$ holds for a square matrix $Q\left(\Delta\overrightarrow{u_{k-1}}(t)\right)$ defined as
\begin{equation}\label{eq015}
Q\left(\Delta\overrightarrow{u_{k-1}}(t)\right)
=I-\frac{\Ds\Delta\overrightarrow{u_{k-1}}(t)\Delta\overrightarrow{u_{k-1}}^{\tp}(t)}
{\Ds\mu_{1}+\mu_{2}+\left\|\Delta\overrightarrow{u_{k-1}}(t)\right\|_{2}^{2}}.
\end{equation}
\end{lem}

\begin{proof}
The proof of this lemma can be given by exploiting the specific symmetric structure of $Q\left(\Delta\overrightarrow{u_{k-1}}(t)\right)$, where the details are given in Appendix \ref{apdx6}.
\end{proof}

With Lemma \ref{lem2}, we show the proof of Theorem \ref{thm2} as follows.

\begin{proof}[Proof of Theorem \ref{thm2}]
It can be seen that (\ref{eq15}) in the Algorithm 1 is equivalently derived from (\ref{eq012}) in Lemma \ref{lem8}. We thus revisit (\ref{eq012}) and can employ (\ref{eq015}) to deduce
\begin{equation}\label{eq020}
\aligned
\overrightarrow{\hat{\theta}_{k,k-1,t}}(t)
&=\frac{\Ds\mu_{1}}{\Ds\mu_1+\mu_2}\Bigg[\overrightarrow{\hat{\theta}_{k-1,k-2,t}}(t)
-\frac{\Ds\Delta\overrightarrow{u_{k-1}}^{\tp}(t)\overrightarrow{\hat{\theta}_{k-1,k-2,t}}(t)\Delta\overrightarrow{u_{k-1}}(t)}
{\Ds\mu_{1}+\mu_{2}+\left\|\Delta\overrightarrow{u_{k-1}}(t)\right\|_{2}^{2}}\Bigg]\\
&~~~+\frac{\Ds\Delta y_{k-1}(t+1)\Delta\overrightarrow{u_{k-1}}(t)}
{\Ds\mu_{1}+\mu_{2}+\left\|\Delta\overrightarrow{u_{k-1}}(t)\right\|_{2}^{2}}\\
&=\left[\frac{\Ds\mu_{1}}{\Ds\mu_1+\mu_2}Q\left(\Delta\overrightarrow{u_{k-1}}(t)\right)\right]\overrightarrow{\hat{\theta}_{k-1,k-2,t}}(t)
+\frac{\Ds\Delta y_{k-1}(t+1)\Delta\overrightarrow{u_{k-1}}(t)}
{\Ds\mu_{1}+\mu_{2}+\left\|\Delta\overrightarrow{u_{k-1}}(t)\right\|_{2}^{2}}
\endaligned
\end{equation}

\noindent where we also insert
\[\aligned
\Delta\overrightarrow{u_{k-1}}^{\tp}(t)&\overrightarrow{\hat{\theta}_{k-1,k-2,t}}(t)\Delta\overrightarrow{u_{k-1}}(t)
=\Delta\overrightarrow{u_{k-1}}(t)\Delta\overrightarrow{u_{k-1}}^{\tp}(t)\overrightarrow{\hat{\theta}_{k-1,k-2,t}}(t).
\endaligned\]

\noindent By combining (\ref{eq8}) and (\ref{eq10}), we can validate
\begin{equation}\label{eq017}
\aligned
\left|\Delta y_{k-1}(t+1)\right|
&=\left|\Delta\overrightarrow{u_{k-1}}^{\tp}(t)\overrightarrow{\theta_{k-1,k-2,t}}(t)\right|\\
&\leq\left\|\Delta\overrightarrow{u_{k-1}}(t)\right\|_{2}\left\|\overrightarrow{\theta_{k-1,k-2,t}}(t)\right\|_{2}\\
&=\left\|\Delta\overrightarrow{u_{k-1}}(t)\right\|_{2}\sqrt{\sum_{i=0}^{t}\theta_{k-1,k-2,t}^2(i)}\\
&\leq\sqrt{t+1}\beta_{\theta}\|\left\|\Delta\overrightarrow{u_{k-1}}(t)\right\|_{2}\\
& \leq\sqrt{T}\beta_{\theta}\left\|\Delta\overrightarrow{u_{k-1}}(t)\right\|_{2}.
\endaligned
\end{equation}

\noindent We further explore (\ref{eq017}) to derive
\begin{equation}\label{eq019}
\aligned
\left\|\frac{\Ds\Delta y_{k-1}(t+1)\Delta\overrightarrow{u_{k-1}}(t)}
{\Ds\mu_{1}+\mu_{2}+\left\|\Delta\overrightarrow{u_{k-1}}(t)\right\|_{2}^{2}}\right\|_{2}
&\leq\frac{\Ds\sqrt{T}\beta_{\theta}\left\|\Delta\overrightarrow{u_{k-1}}(t)\right\|_{2}^{2}}
{\Ds\mu_{1}+\mu_{2}+\left\|\Delta\overrightarrow{u_{k-1}}(t)\right\|_{2}^{2}}\\
&\leq\sqrt{T}\beta_{\theta}.
\endaligned
\end{equation}

\noindent With Lemma \ref{lem2}, we consider (\ref{eq019}) for (\ref{eq020}) and can obtain
\begin{equation}\label{eq041}
\aligned
\left\|\overrightarrow{\hat{\theta}_{k,k-1,t}}(t)\right\|_{2}
&\leq\left\|\frac{\Ds\mu_{1}}{\Ds\mu_1+\mu_2}Q\left(\Delta\overrightarrow{u_{k-1}}(t)\right)\right\|_{2}
\left\|\overrightarrow{\hat{\theta}_{k-1,k-2,t}}(t)\right\|_{2}
+\left\|\frac{\Ds\Delta y_{k-1}(t+1)\Delta\overrightarrow{u_{k-1}}(t)}
{\Ds\mu_{1}+\mu_{2}+\left\|\Delta\overrightarrow{u_{k-1}}(t)\right\|_{2}^{2}}\right\|_{2}\\
&\leq\frac{\Ds\mu_{1}}{\Ds\mu_{1}+\mu_{2}}\left\|\overrightarrow{\hat{\theta}_{k-1,k-2,t}}(t)\right\|_{2}
+\sqrt{T}\beta_{\theta}
\endaligned
\end{equation}

\noindent which can be adopted to yield
\begin{equation}\label{eq021}
\aligned
\left\|\overrightarrow{\hat{\theta}_{k,k-1,t}}(t)\right\|_{2}
&\leq\left(\frac{\Ds\mu_{1}}{\Ds\mu_{1}+\mu_{2}}\right)^{k-1}\left\|\overrightarrow{\hat{\theta}_{1,0,t}}(t)\right\|_{2}
+\sum_{i=0}^{k-2}\left(\frac{\Ds\mu_{1}}{\Ds\mu_{1}+\mu_{2}}\right)^{i}\sqrt{T}\beta_{\theta}.
\endaligned
\end{equation}

\noindent Due to $\mu_{1}/\left(\mu_{1}+\mu_{2}\right)<1$, we can verify with (\ref{eq021}) that
\begin{equation}\label{eq022}
\aligned
\left\|\overrightarrow{\hat{\theta}_{k,k-1,t}}(t)\right\|_{2}
&\leq\left\|\overrightarrow{\hat{\theta}_{1,0,t}}(t)\right\|_{2}
+\frac{\Ds\mu_{1}+\mu_{2}}{\Ds\mu_{2}}\sqrt{T}\beta_{\theta}\\
&\leq\beta_{\hat{\theta}},\quad\forall t\in\Z_{T-1},\forall k\in\Z
\endaligned
\end{equation}

\noindent where
\[\beta_{\hat{\theta}}
=\max_{t\in\Z_{T-1}}\left\|\overrightarrow{\hat{\theta}_{1,0,t}}(t)\right\|_{2}
+\frac{\Ds\mu_{1}+\mu_{2}}{\Ds\mu_{2}}\sqrt{T}\beta_{\theta}.
\]

\noindent Since $\left|\th_{k,k-1,t}(i)\right|\leq\left\|\overrightarrow{\hat{\theta}_{k,k-1,t}}(t)\right\|_{2}$, $\forall i\in\Z_{t}$ holds, (\ref{eq18}) follows as a direct consequence of (\ref{eq022}). In particular, we can develop (\ref{eq018}) by also considering that (\ref{eq010}) ensures $\th_{k,k-1,t}(t)\geq\epsilon$.
\end{proof}

\begin{rem}\label{rem5}
It is worth emphasizing that (\ref{eq020}) essentially gives a nonrepetitive system with respect to iteration because of the system matrix $\mu_{1}Q\left(\Delta\overrightarrow{u_{k-1}}(t)\right)/\left(\mu_{1}+\mu_{2}\right)$ (see also \cite{mm:17,mz:19}). Despite this issue, we can develop a strict contraction mapping condition as
\begin{equation}\label{eq040}
\left\|\frac{\Ds\mu_{1}}{\Ds\mu_1+\mu_2}Q\left(\Delta\overrightarrow{u_{k-1}}(t)\right)\right\|_{2}
\leq\frac{\Ds\mu_{1}}{\Ds\mu_1+\mu_2}<1,\quad\forall t\in\Z_{T-1},\forall k\geq2
\end{equation}

\noindent and, hence, we can directly implement the CM-based approach to the boundedness analysis for the estimation $\th_{k,k-1,t}(i)$ of the uncertain parameter $\theta_{k,k-1,t}(i)$, $\forall i\in\Z_{t}$, $\forall t\in\Z_{T-1}$, $\forall k\in\Z$. In fact, such a benefit is because of the optimization-based design result of Lemma \ref{lem8}, which can no longer be gained for $\mu_{2}=0$. We can consequently see that we may improve the boundedness analysis method used in, e.g., \cite{chhj:15,chjh:18}.
\end{rem}

\subsection{Convergence of Optimization-Based Adaptive ILC}

We proceed to explore the system performances of (\ref{eq01}) under the Algorithm 1 of optimization-based adaptive ILC, including the boundedness of the system trajectories and the convergence of the tracking error. We thus revisit (\ref{eq25}) that essentially shows a nonrepetitive higher-order ILC process regarding the tracking error. To overcome the effect of higher-order dynamics on ILC, we resort to a lifting technique to reformulate (\ref{eq25}) as
\begin{equation}\label{eq29}
\overrightarrow{e_{k}}(t+1)=P_{k}(t)\overrightarrow{e_{k-1}}(t+1)+\overrightarrow{\kappa_{k}}(t)
\end{equation}

\noindent where $\overrightarrow{e_{k}}(t+1)$ and $\overrightarrow{\kappa_{k}}(t)$ are two vectors defined by
\begin{equation}\label{eq28}
\aligned
\overrightarrow{e_{k}}(t+1)&=\left[e_{k}(t+1),e_{k-1}(t+1),\cdots,e_{k-m+2}(t+1)\right]^{\tp}\in\R^{m-1}\\
\overrightarrow{\kappa_{k}}(t)&=\left[\kappa_{k}(t),0,\cdots,0\right]^{\tp}\in\R^{m-1}
\endaligned
\end{equation}

\noindent and $P_{k}(t)\in\R^{(m-1)\times(m-1)}$ is a correspondingly induced matrix in the form of
\begin{equation}\label{eq30}
\aligned
P_{k}(t)&=\begin{bmatrix}
p_{1,k}(t) & p_{2,k}(t) & \cdots & \cdots & p_{m-1,k}(t)\\
1 & 0 & \cdots & \cdots & 0 \\
0 &  \ddots  & \ddots & \ddots & \vdots  \\
\vdots& \ddots & \ddots& \ddots & \vdots \\
0 & \cdots & 0 & 1 & 0
\end{bmatrix}~~\hbox{with}\\
p_{1,k}(t)
&=1-\frac{\Ds\left(\gamma_{1}^{2}+\gamma_{1}\gamma_{2}\right)\theta_{k,k-1,t}(t)\th_{k,k-1,t}(t)}
{\Ds\lambda+\gamma_{1}^{2}\th^{2}_{k,k-1,t}(t)}\\
p_{i,k}(t)
&=-\frac{\Ds\gamma_{1}\gamma_{i+1}\theta_{k,k-1,t}(t)\th_{k,k-1,t}(t)}{\Ds\lambda+\gamma_{1}^{2}\th^{2}_{k,k-1,t}(t)},\quad i=2,3,\cdots,m-1.
\endaligned
\end{equation}

For (\ref{eq29}), we can develop a convergence result by leveraging the nonrepetitive ILC results of, e.g., \cite{mm:17}.

\begin{lem}\label{lem10}
For (\ref{eq29}) over any $t\in\Z_{T-1}$, if 
\begin{enumerate}
\item[(C)]
there exist some iteration sequence $\left\{\omega_{s}(t)\right\}_{s=0}^{\infty}$ and some finite positive integer $\chi(t)$, satisfying $\omega_{0}(t)=1$ and $0<\omega_{s+1}(t)-\omega_{s}(t)\leq\chi(t)$, such that
\begin{equation*}
\left\|\prod_{k=\omega_{s}(t)}^{\omega_{s+1}(t)-1}P_{k}(t)\right\|_{\infty}\leq\eta<1,\quad\forall s\in\Z_{+}
\end{equation*}
\end{enumerate}

\noindent then the following two results hold:
\begin{enumerate}
\item
$\overrightarrow{e_{k}}(t+1)$ is bounded (that is, $\sup_{k\in\Z_{+}}\left|\overrightarrow{e_{k}}(t+1)\right|\leq\beta_{\overrightarrow{e}}(t)$ for some finite bound $\beta_{\overrightarrow{e}}(t)>0$), provided that $\overrightarrow{\kappa_{k}}(t)$ is bounded (that is, $\sup_{k\in\Z_{+}}\left|\overrightarrow{\kappa_{k}}(t)\right|\leq\beta_{\overrightarrow{\kappa}}(t)$ for some finite bound $\beta_{\overrightarrow{\kappa}}(t)>0$);

\item
$\lim_{k\to\infty}\overrightarrow{e_{k}}(t+1)=0$, provided that $\lim_{k\to\infty}\overrightarrow{\kappa_{k}}(t)=0$.
\end{enumerate}
\end{lem}

\begin{proof}
The two results in this lemma can be shown by utilizing the results i) and ii) in \cite[Lemma 2]{mm:17} to (\ref{eq29}), respectively.
\end{proof}

For the condition (C) in Lemma \ref{lem10}, we can verify
\begin{equation}\label{eq023}
\left\|\prod_{k=\omega_{s}(t)}^{\omega_{s+1}(t)-1}P_{k}(t)\right\|_{\infty}
\leq\left\|\prod_{k=\omega_{s}(t)}^{\omega_{s+1}(t)-1}\left|P_{k}(t)\right|\right\|_{\infty},\quad\forall t\in\Z_{+},\forall k\in\Z
\end{equation}

\noindent in which $\left|P_{k}(t)\right|$, compared with $P_{k}(t)$, becomes a nonnegative matrix given by
\[
\left|P_{k}(t)\right|
=\begin{bmatrix}
\left|p_{1,k}(t)\right| & \left|p_{2,k}(t)\right| & \cdots & \cdots & \left|p_{m-1,k}(t)\right|\\
1 & 0 & \cdots & \cdots & 0 \\
0 &  \ddots  & \ddots & \ddots & \vdots  \\
\vdots& \ddots & \ddots& \ddots & \vdots \\
0 & \cdots & 0 & 1 & 0
\end{bmatrix}.
\]

\noindent We explore the fact (\ref{eq023}) based on the properties of nonnegative matrices, together with using the convergence result of Lemma \ref{lem10}, to establish a convergence result with respect to the tracking error satisfying (\ref{eq25}).

\begin{lem}\label{lem4}
For (\ref{eq25}) over any $t\in\Z_{T-1}$, if
\begin{equation}\label{eq27}
\aligned
&\left |1-\frac{\Ds\left(\gamma_{1}^{2}+\gamma_{1}\gamma_{2}\right)\theta_{k,k-1,t}(t)\th_{k,k-1,t}(t)}
{\Ds\lambda+\gamma_{1}^{2}\th^{2}_{k,k-1,t}(t)}\right|
+\sum_{i=3}^{m}\left|\frac{\Ds\gamma_{1}\gamma_{i}\theta_{k,k-1,t}(t)\th_{k,k-1,t}(t)}{\Ds\lambda+\gamma_{1}^{2}\th^{2}_{k,k-1,t}(t)}\right|
\leq\zeta<1,\quad\forall k\in\Z
\endaligned
\end{equation}

\noindent then the following two results hold:
\begin{enumerate}
\item
$e_{k}(t+1)$ is bounded (namely, $\sup_{k\in\Z_{+}}\left|e_{k}(t+1)\right|\leq\beta_{e}(t)$ for some finite bound $\beta_{e}(t)>0$), provided that $\kappa_{k}(t)$ is bounded (namely, $\sup_{k\in\Z_{+}}\left|\kappa_{k}(t)\right|\leq\beta_{\kappa}(t)$ for some finite bound $\beta_{\kappa}(t)>0$);

\item
$\lim_{k\to\infty}e_{k}(t+1)=0$, provided that $\lim_{k\to\infty}\kappa_{k}(t)=0$.
\end{enumerate}
\end{lem}

\begin{proof}
This lemma can be developed via a nonnegative matrix-based analysis approach and by noting the definitions (\ref{eq28}) and (\ref{eq30}). For the proof details, see Appendix \ref{apdx5}.
\end{proof}

Even though Lemma \ref{lem4} may help to achieve the convergence analysis of the tracking error, it is no longer applicable for the boundedness analysis of the system trajectories. We thus resort to the DDA approach to ILC and exploit the dynamic evolution of input along the iteration axis to implement the boundedness analysis in the presence of nonrepetitive uncertainties (see also \cite{mm:17,mz:19}). Toward this end, we rewrite (\ref{eq17}) as
\begin{equation}\label{eq20}
\aligned
u_{k}(t)
&=u_{k-1}(t)+\frac{\Ds\gamma_{1}\th_{k,k-1,t}(t)\sum_{i=3}^{m}\gamma_{i}e_{k-i+1}(t+1)}{\Ds\lambda+\gamma_{1}^{2}\th^{2}_{k,k-1,t}(t)}
-\frac{\Ds\gamma_{1}^{2}\th_{k,k-1,t}(t)\sum_{i=0}^{t-1}\th_{k,k-1,t}(i)\Delta u_k(i)}{\Ds\lambda+\gamma_{1}^{2}\th^{2}_{k,k-1,t}(t)}\\
&~~~+\frac{\Ds\left(\gamma_{1}^{2}+\gamma_{1}\gamma_{2}\right)\th_{k,k-1,t}(t)\left[y_d(t+1)-y_{k-1}(t+1)\right]}
{\Ds\lambda+\gamma_{1}^{2}\th^{2}_{k,k-1,t}(t)}.
\endaligned
\end{equation}

\noindent As a direct application of (\ref{eq09}) for the initial iteration (i.e., $j=0$) and the $(k-1)$th iteration (i.e., $i=k-1$), we can derive
\begin{equation}\label{eq21}
\aligned
y_{k-1}(t+1)
&=y_{0}(t+1)+\sum_{i=0}^{t}\theta_{k-1,0,t}(i)\left[u_{k-1}(i)-u_{0}(i)\right]\\
&=\theta_{k-1,0,t}(t)u_{k-1}(t)+y_0(t+1)
+\sum_{i=0}^{t-1}\theta_{k-1,0,t}(i)u_{k-1}(i)-\sum_{i=0}^{t}\theta_{k-1,0,t}(i)u_0(i).
\endaligned
\end{equation}

\noindent By substituting (\ref{eq21}) into (\ref{eq20}), we can obtain
\begin{equation}\label{eq22}
\aligned
u_{k}(t)
&=\left[1-\frac{\Ds\left(\gamma_{1}^{2}+\gamma_{1}\gamma_{2}\right)\th_{k,k-1,t}(t)\theta_{k-1,0,t}(t)}
{\Ds\lambda+\gamma_{1}^{2}\th^{2}_{k,k-1,t}(t)}\right]u_{k-1}(t)
+\psi_{k}(t)
\endaligned
\end{equation}

\noindent where $\psi_{k}(t)$ is a driving signal given by
\begin{equation}\label{eq23}
\aligned
\psi_{k}(t)
&=\frac{\Ds\gamma_{1}\th_{k,k-1,t}(t)}{\Ds\lambda+\gamma_{1}^{2}\th^{2}_{k,k-1,t}(t)}
\Bigg[\left(\gamma_{1}+\gamma_{2}\right)\sum_{i=0}^{t}\theta_{k-1,0,t}(i)u_{0}(i)
-\big(\gamma_{1}
+\gamma_{2}\big)\sum_{i=0}^{t-1}\theta_{k-1,0,t}(i)u_{k-1}(i)\\
&~~~-\gamma_{1}\sum_{i=0}^{t-1}\th_{k,k-1,t}(i)\Delta u_{k}(i)+\left(\gamma_{1}+\gamma_{2}\right)e_{0}(t+1)+\sum_{i=3}^{m}\gamma_{i}e_{k-i+1}(t+1)\Bigg].
\endaligned
\end{equation}

With (\ref{eq22}), we propose a lemma for the boundedness of the input $u_{k}(t)$ with respect to any bounded driving signal $\psi_{k}(t)$.

\begin{lem}\label{lem3}
For (\ref{eq22}) over any $t \in \Z_{T-1}$, if
\begin{equation}\label{eq24}
\left|1-\frac{\Ds\left(\gamma_{1}^{2}+\gamma_{1}\gamma_{2}\right)\th_{k,k-1,t}(t)\theta_{k-1,0,t}(t)}
{\Ds\lambda+\gamma_{1}^{2}\th^{2}_{k,k-1,t}(t)}\right|\leq\phi<1,\quad\forall k\in\Z
\end{equation}

\noindent then $u_{k}(t)$ is ensured to be bounded such that $\sup_{k\in\Z_{+}}\left|u_{k}(t)\right|\leq\beta_{u}(t)$ holds for some finite bound $\beta_{u}(t)>0$, provided that $\psi_{k}(t)$ is bounded (i.e., $\sup_{k\in\Z_{+}}\left|\psi_{k}(t)\right|\leq\beta_{\psi}(t)$ for some finite bound $\beta_{\psi}(t)>0$).
\end{lem}

\begin{proof}
A consequence of the result i) of \cite[Lemma 2]{mm:17}.
\end{proof}

Based on the analysis results of Lemmas \ref{lem4} and \ref{lem3}, we present the following theorem to develop tracking results for uncertain nonlinear systems using the Algorithm 1 of optimization-based adaptive ILC.

\begin{thm}\label{thm1}
For the nonlinear system (\ref{eq01}), let the Assumption (A1) be satisfied. If the Algorithm 1 is applied with 
\begin{equation}\label{eq19}
\gamma_{1}+\gamma_{2}>\sum_{i=3}^{m}\gamma_{i},\quad
\lambda>\left(\gamma_{1}^{2}+\gamma_{1}\gamma_{2}\right)\beta_{\overline{f}}\beta_{\th}
\end{equation}

\noindent then both boundedness and convergence of optimization-based adaptive ILC can be ensured, namely,
\begin{enumerate}
\item
the boundedness of the input and output trajectories can be guaranteed during the ILC process such that
\begin{equation}\label{eq039}
\aligned
\left|u_{k}(t)\right|&\leq\beta_{u},\quad \forall t\in\Z_{T-1},\forall k\in\Z_{+}\\
\left|y_{k}(t)\right|&\leq\beta_{y},\quad \forall t\in\Z_{T},\forall k\in\Z_{+}
\endaligned
\end{equation}

\noindent for some finite bounds $\beta_{u}>0$ and $\beta_{y}>0$;

\item
the perfect tracking objective (\ref{eq07}) of ILC can be achieved.
\end{enumerate}
\end{thm}

\begin{rem}\label{rem6}
From Theorem \ref{thm1}, it can be seen that the Algorithm 1 is effective in accomplishing the perfect output tracking tasks of ILC, together with guaranteeing the boundedness of all the system trajectories, even in the presence of unknown nonlinear time-varying dynamics. Moreover, it is worth highlighting that Theorem \ref{thm1} actually provides a class of data-driven ILC results because the implementation of the nonlinear system (\ref{eq01}) under the Algorithm 1 leverages only the input and output data. Also, we need quite limited estimation knowledge of (\ref{eq01}) to establish Theorem \ref{thm1}, as well as to gain Theorem \ref{thm2}. Similar contributions have been made in, e.g., \cite{chhj:15,chjh:18}, which however employ the eigenvalue-based CM approach to the convergence analysis of ILC and can no longer apply to the analysis of our results. This can be clearly seen from (\ref{eq25}) and (\ref{eq22}) that yield nonrepetitive ILC processes and make the eigenvalue analysis not applicable for their convergence analysis any longer (see also \cite{mjd:15}).
\end{rem}
\begin{rem}\label{rem7}
Another issue worth noticing is the interdependent relation between (\ref{eq25}) and (\ref{eq26}) for the tracking error dynamics and (\ref{eq22}) and (\ref{eq23}) for the input dynamics. This naturally leads to that the convergence analysis of the tracking error and the boundedness analysis of the input are interdependent with each other. Thus, the CM-based approach to ILC can not be adopted to develop Theorem \ref{thm1}, which motivates us to implement a DDA approach. A benefit of employing a DDA approach is to make the selection condition (\ref{eq19}) independent of the length $T$ of the learning time interval. It consequently improves the selection condition used in, e.g., \cite{chjh:18} that depends heavily on the length of the learning time interval.
\end{rem}
\begin{rem}\label{rem8}
In particular, the results of Theorem \ref{thm1} as well as of Theorem \ref{thm2} are applicable for nonlinear systems that are time-invariant and, hence, can generalize the existing optimization-based adaptive ILC results of, e.g., \cite{ch:07}-\cite{chhj:18}. Another special case is to consider the first-order optimization-based adaptive ILC (i.e., $m=1$), one of the most considered cases of adaptive ILC for nonlinear ILC. In this special case, our derived results still work effectively, where we only need to take $\gamma_{i}=0$, $i=2$, $3$, $\cdots$, $m$ in our design and analysis. For example, when $m=1$, the only modification of the Algorithm 1 is that (\ref{eq17}) becomes
\begin{equation*}\label{}
\aligned
u_{k}(t)
&=u_{k-1}(t)
+\frac{\Ds\gamma_{1}^{2}\th_{k,k-1,t}(t)}{\Ds\lambda+\gamma_{1}^{2}\th^{2}_{k,k-1,t}(t)}
\Bigg\{e_{k-1}(t+1)
-\sum_{i=0}^{t-1}\th_{k,k-1,t}(i)\left[u_{k}(i)-u_{k-1}(i)\right]\Bigg\}
\endaligned
\end{equation*}

\noindent and Theorems \ref{thm2} and \ref{thm1} hold, for which the selection condition (\ref{eq19}) collapses into $\lambda>\gamma_{1}^{2}\beta_{\overline{f}}\beta_{\th}$.
\end{rem}

Although Lemmas \ref{lem4} and \ref{lem3} show preliminary analysis results for the development of Theorem \ref{thm1}, they resort to two different conditions (\ref{eq27}) and (\ref{eq24}). To overcome this issue, we introduce a helpful lemma to disclose the relations among the conditions (\ref{eq27}), (\ref{eq24}) and (\ref{eq19}).

\begin{lem}\label{lem6}
For the nonlinear system (\ref{eq01}) under the Assumption (A1), if the condition (\ref{eq19}) is satisfied, then both conditions (\ref{eq27}) and (\ref{eq24}) can be simultaneously guaranteed.
\end{lem}

\begin{proof}
This lemma can be proved with the boundedness results of Lemma \ref{lem1} and Theorem \ref{thm2}, where the proof details are given in Appendix \ref{apdx2}.
\end{proof}

Now, by utilizing Lemmas \ref{lem4}--\ref{lem6}, we are in position to present the proof of Theorem \ref{thm1}, for which a DDA approach instead of the eigenvalue-based analysis approach to ILC is implemented.

\begin{proof}[Proof of Theorem \ref{thm1}]
It follows by Lemma \ref{lem6} that the selection condition (\ref{eq19}) in this theorem ensures the validity of Lemmas \ref{lem4} and \ref{lem3}. Then we perform induction over $t\in\Z_{T-1}$ to complete this proof with two steps.

{\it Step i): Let $t=0$, and then we prove that $\lim_{k\to\infty}e_{k}(1)=0$, $\lim_{k\to\infty}\Delta u_{k}(0)=0$, and $\sup_{k\in\Z_{+}}\left|u_{k}(0)\right|\leq\beta_{u}(0)$ for some finite bound $\beta_{u}(0)>0$.}

From (\ref{eq26}), it follows $\kappa_{k}(0)=0$, $\forall k\in\Z_{+}$. We hence consider Lemma \ref{lem4} for (\ref{eq25}) and can obtain
\begin{equation}\label{eq36}
\aligned
\sup_{k\in\Z_{+}}\left|e_{k}(1)\right|\leq\beta_{e}(0)
~~\hbox{and}~~\lim_{k \to \infty}e_{k}(1)=0
\endaligned
\end{equation}

\noindent for some finite bound $\beta_{e}(0)>0$. From (\ref{eq17}), we can deduce
\begin{equation*}\label{}
\aligned
\Delta u_{k}(0)
&=\frac{\Ds\gamma_{1}\th_{k,k-1,0}(0)}{\Ds\lambda+\gamma_{1}^{2}\th^{2}_{k,k-1,0}(0)}
\left[\gamma_{1}e_{k-1}(1)+\sum_{i=2}^{m}\gamma_{i}e_{k-i+1}(1)\right]
\endaligned
\end{equation*}

\noindent which, together with (\ref{eq018}) and (\ref{eq36}), leads to $\lim_{k\to\infty}\Delta u_{k}(0)=0$. In addition, we know from (\ref{eq23}) that
\begin{equation*}\label{}
\aligned
\psi_{k}(0)
&=\frac{\Ds\gamma_{1}\th_{k,k-1,0}(0)}{\Ds\lambda+\gamma_{1}^{2}\th^{2}_{k,k-1,0}(0)}
\Bigg[\left(\gamma_{1}+\gamma_{2}\right)\theta_{k-1,0,0}(0)u_{0}(0)
+\left(\gamma_{1}+\gamma_{2}\right)e_{0}(1)+\sum_{i=3}^{m}\gamma_{i}e_{k-i+1}(1)\Bigg]
\endaligned
\end{equation*}

\noindent and then by the boundedness results of Lemma \ref{lem1} and Theorem \ref{thm2}, we can derive
\begin{equation}\label{eq37}
\aligned
\left|\psi_{k}(0)\right|
&\leq\frac{\Ds\gamma_{1}\beta_{\th}}{\Ds\lambda+\gamma_{1}^{2}\varepsilon^{2}}
\left[\left(\gamma_{1}+\gamma_{2}\right)\beta_{\theta}\left|u_{0}(0)\right|
+\sum_{i=1}^{m}\gamma_{i}\beta_{e}(0)\right]\\%
&\triangleq\beta_{\psi}(0).
\endaligned
\end{equation}

\noindent With (\ref{eq37}), the use of Lemma \ref{lem3} yields $\sup_{k\in\Z_{+}}\left|u_{k}(0)\right|\leq\beta_{u}(0)$ for some finite bound $\beta_{u}(0)>0$.

{\it Step ii): For $t=0$, $1$, $\cdots$, $N-1$ with any given $N\in\Z_{T-1}$, let $\lim_{k\to\infty}e_{k}(t+1)=0$, $\lim_{k\to\infty}\Delta u_{k}(t)=0$, and $\sup_{k\in\Z_{+}}\left|u_{k}(t)\right|\leq\beta_{u}(t)$ for some finite bound $\beta_{u}(t)>0$. Then, for $t=N$, we will prove that the hypothesis made for the two convergence results and one boundedness result also holds.}

With the hypothesis made for the time steps $0$, $1$, $\cdots$, $N-1$ in {\it Step ii)} and by applying the boundedness results of Lemma \ref{lem1} and Theorem \ref{thm2}, we can employ (\ref{eq26}) to verify
\begin{equation}\label{eq025}
\aligned
\left|\kappa_{k}(N)\right|
&=\Bigg|\frac{\Ds\gamma_{1}^{2}\theta_{k,k-1,N}(N)\th_{k,k-1,N}(N)}{\Ds\lambda+\gamma_{1}^{2}\th^{2}_{k,k-1,N}(N)}
\sum_{i=0}^{N-1}\th_{k,k-1,N}(i)\Delta u_{k}(i)
-\sum_{i=0}^{N-1}\theta_{k,k-1,N}(i)\Delta u_{k}(i)\Bigg|\\
&\leq2\beta_{\theta}
\left(1+\frac{\Ds\gamma_{1}^{2}\beta_{\th}^{2}}{\Ds\lambda+\gamma_{1}^{2}\varepsilon^{2}}\right)
\sum_{i=0}^{N-1}\beta_{u}(i)\\
&\triangleq\beta_{\kappa}(N),\quad\forall k\in\Z_{+}
\endaligned
\end{equation}

\noindent and
\begin{equation}\label{eq026}
\aligned
\left|\kappa_{k}(N)\right|
&\leq\beta_{\theta}
\left(1+\frac{\Ds\gamma_{1}^{2}\beta_{\th}^{2}}{\Ds\lambda+\gamma_{1}^{2}\varepsilon^{2}}\right)
\sum_{i=0}^{N-1}\left|\Delta u_{k}(i)\right|\\
&\to0,\quad\hbox{as}~k\to\infty.
\endaligned
\end{equation}

\noindent We then leverage (\ref{eq025}) and (\ref{eq026}) and apply Lemma \ref{lem4} to deduce
\begin{equation}\label{eq38}
\sup_{k\in\Z_{+}}\left|e_{k}(N+1)\right|\leq\beta_{e}(N)
~~\hbox{and}~~\lim_{k \to \infty}e_{k}(N+1)=0
\end{equation}

\noindent for some finite bound $\beta_{e}(N)>0$. Since we can use (\ref{eq17}) to get
\begin{equation*}
\aligned
\Delta u_{k}(N)&=
-\frac{\Ds\gamma_{1}^{2}\th_{k,k-1,N}(N)}{\Ds\lambda+\gamma_{1}^{2}\th^{2}_{k,k-1,N}(N)}
\sum_{i=0}^{N-1}\th_{k,k-1,N}(i)\Delta u_{k}(i)\\
&~~~+\frac{\Ds\gamma_{1}\th_{k,k-1,N}(N)}{\Ds\lambda+\gamma_{1}^{2}\th^{2}_{k,k-1,N}(N)}
\Bigg[\gamma_{1}e_{k-1}(N+1)
+\sum_{i=2}^{m}\gamma_{i}e_{k-i+1}(N+1)\Bigg]
\endaligned
\end{equation*}

\noindent we follow the same lines as (\ref{eq026}) and insert (\ref{eq38}) to derive
\begin{equation*}
\aligned
\left|\Delta u_{k}(N)\right|
&\leq\frac{\Ds\gamma_{1}^{2}\beta_{\th}^{2}}{\Ds\lambda+\gamma_{1}^{2}\varepsilon^{2}}
\sum_{i=0}^{N-1}\left|\Delta u_{k}(i)\right|
+\frac{\Ds\gamma_{1}\beta_{\th}}{\Ds\lambda+\gamma_{1}^{2}\varepsilon^{2}}
\Bigg[\gamma_{1}\left|e_{k-1}(N+1)\right|
+\sum_{i=2}^{m}\gamma_{i}\left|e_{k-i+1}(N+1)\right|\Bigg]\\
&\to0,\quad\hbox{as}~k\to\infty
\endaligned
\end{equation*}

\noindent which implies $\lim_{k\to\infty}\Delta u_{k}(N)=0$. From (\ref{eq23}), we can obtain
\[
\aligned
\psi_{k}(N)
&=\frac{\Ds\gamma_{1}\th_{k,k-1,N}(N)}{\Ds\lambda+\gamma_{1}^{2}\th^{2}_{k,k-1,N}(N)}
\Bigg[\left(\gamma_{1}+\gamma_{2}\right)\sum_{i=0}^{N}\theta_{k-1,0,N}(i)u_{0}(i)
-\big(\gamma_{1}
+\gamma_{2}\big)\sum_{i=0}^{N-1}\theta_{k-1,0,N}(i)u_{k-1}(i)\\
&~~~-\gamma_{1}\sum_{i=0}^{N-1}\th_{k,k-1,N}(i)\Delta u_{k}(i)+\left(\gamma_{1}+\gamma_{2}\right)e_{0}(N+1)+\sum_{i=3}^{m}\gamma_{i}e_{k-i+1}(N+1)\Bigg]
\endaligned
\]

\noindent with which we can validate
\begin{equation}\label{eq39}
\aligned
\left|\psi_{k}(N)\right|
&\leq\frac{\Ds\gamma_{1}\beta_{\th}}{\Ds\lambda+\gamma_{1}^{2}\varepsilon^{2}}
\Bigg[\left(\gamma_{1}+\gamma_{2}\right)\beta_{\theta}\sum_{i=0}^{N}\left|u_{0}(i)\right|
+\left(\gamma_{1}\beta_{\theta}+\gamma_{2}\beta_{\theta}+2\gamma_{1}\beta_{\th}\right)
\sum_{i=0}^{N-1}\beta_{u}(i)
+\sum_{i=1}^{m}\gamma_{i}\beta_{e}(N)\Bigg]\\
&\triangleq\beta_{\psi}(N).
\endaligned
\end{equation}

\noindent Based on (\ref{eq39}), we consider Lemma \ref{lem3} for (\ref{eq22}) and can develop $\sup_{k\in\Z_{+}}\left|u_{k}(N)\right|\leq\beta_{u}(N)$ for some finite bound $\beta_{u}(N)>0$. We can thus conclude that the hypothesis made for $t=0$, $1$, $\cdots$, $N-1$ in this step also holds for $t=N$.

By induction based on the analysis of the above steps i) and ii), we can arrive at
\begin{equation}\label{eq027}
\sup_{k\in\Z_{+}}\left|u_{k}(t)\right|\leq\beta_{u}(t)
~\hbox{and}~\lim_{k\to\infty}e_{k}(t+1)=0,\quad\forall t\in\Z_{T-1}
\end{equation}

\noindent with which we can further employ Lemma \ref{lem4} to get
\begin{equation}\label{eq028}
\sup_{k\in\Z_{+}}\left|e_{k}(t+1)\right|\leq\beta_{e}(t),\quad\forall t\in\Z_{T-1}.
\end{equation}

\noindent The use of (\ref{eq028}) yields $\sup_{k\in\Z_{+}}\left|y_{k}(t+1)\right|\leq\beta_{e}(t)+\left|y_{d}(t+1)\right|$, $\forall t\in\Z_{T-1}$, which together with (\ref{eq027}) leads to $\sup_{k\in\Z_{+}}\left|u_{k}(t)\right|\leq\beta_{u}$, $\forall t\in\Z_{T-1}$ and $\sup_{k\in\Z_{+}}\left|y_{k}(t)\right|\leq\beta_{y}$, $\forall t\in\Z_{T}$ by taking
\[\aligned
\beta_{u}&=\max_{t\in\Z_{T-1}}\beta_{u}(t)\\
\beta_{y}&=\max\left\{\left|y_{0}\right|,\max_{t\in\Z_{T-1}}\left\{\beta_{e}(t)+\left|y_{d}(t+1)\right|\right\}\right\}.
\endaligned
\]

\noindent We can also derive from (\ref{eq027}) that the perfect tracking objective (\ref{eq07}) holds. The proof of Theorem \ref{thm1} is complete.
\end{proof}

%
%
%

\section{Robustness V.S. Nonrepetitive Uncertainties}\label{sec5}

In this section, we contribute to involving the robust analysis of optimization-based adaptive ILC, regardless of nonrepetitive uncertainties arising from iteration-dependent disturbances and initial shifts. We thus consider the following nonlinear system:
\begin{equation}\label{eq029}
\aligned
y_{k}(t+1)=&f\left(y_{k}(t),\cdots,y_{k}(t-l),u_{k}(t),\cdots,u_{k}(t-n),t\right)+w_{k}(t)\\
\hbox{with}~
&y_{k}(i)
=\left\{\aligned
&0,~~~~~~~i<0\\
&y_{0}+\delta_{k},~i=0
\endaligned\right.~\hbox{and}~
u_k(i)=0,~i<0
\endaligned
\end{equation}

\noindent where, by contrast to (\ref{eq01}), $w_{k}(t)$ and $\delta_{k}$ denote the nonrepetitive disturbance and initial shift, respectively. Due to the presence of nonrepetitive uncertainties, the perfect tracking task (\ref{eq07}) may no longer be achieved in general, and instead a robust tracking task is usually considered of practical importance such that the tracking error can be decreased to a small neighborhood of the origin with increasing iterations, namely,
\begin{equation}\label{eq033}
\limsup_{k\to\infty}\left|e_{k}(t+1)\right|\leq\beta_{e_{\sup}}(t),\quad\forall t\in\Z_{T-1}
\end{equation}

\noindent where $\beta_{e_{\sup}}(t)>0$ is a small bound that depends continuously on those of the nonrepetitive uncertainties.

To implement the robust ILC task, we impose an assumption on the boundedness of nonrepetitive uncertainties.
\begin{enumerate}
\item [(A2)]
Let $w_{k}(t)$ and $\delta_{k}$ be bounded such that
\begin{equation}\label{eq034}
\aligned
\left|w_{k}(t)\right|
&\leq\beta_{w}(t),&\forall t&\in\Z_{T-1},\forall k\in\Z_{+}\\
\left|\delta_{k}\right|
&\leq\beta_{\delta},&\forall k&\in\Z_{+}
\endaligned
\end{equation}

\noindent for some finite bounds $\beta_{w}(t)>0$ and $\beta_{\delta}>0$.
\end{enumerate}

Note that in the ILC literature, (A2) is a common considered assumption for the class of nonrepetitive uncertainties because it can be generally acceptable in many practical situations (see, e.g., \cite{mm:17,mz:19}). Of particular note is to ensure the convergence of nonrepetitive uncertainties such that
\begin{equation}\label{eq035}
\lim_{k\to\infty}\left[w_{k}(t)-w_{k-1}(t)\right]
=0,\forall t\in\Z_{T-1},\quad
\lim_{k\to\infty}\left(\delta_{k}-\delta_{k-1}\right)
=0
\end{equation}

\noindent which may be considered as an additional requirement of (A2) for the accomplishment of the perfect tracking task (\ref{eq07}), despite the presence of nonrepetitive uncertainties. To proceed, we aim at discussing the influence of nonrepetitive uncertainties on the ILC process, for which the following helpful lemma is given to identify the roles of nonrepetitive uncertainties in the extended dynamical linearization for the nonlinear system (\ref{eq01}).

\begin{lem}\label{lem11}
If the Assumption (A1) is satisfied for the nonlinear system (\ref{eq029}), then an extended dynamical linearization for (\ref{eq029}) can be given by
\begin{equation}\label{eq036}
\aligned
\begin{bmatrix}
  y_{i}(1) \\
  y_{i}(2) \\
  \vdots \\
  y_{i}(T)
\end{bmatrix}-\begin{bmatrix}
  y_{j}(1) \\
  y_{j}(2) \\
  \vdots \\
  y_{j}(T)
\end{bmatrix}
&=\Theta_{i,j}\left(\begin{bmatrix}
  u_{i}(0) \\
  u_{i}(1) \\
  \vdots \\
  u_{i}(T-1)
\end{bmatrix}
-\begin{bmatrix}
  u_{j}(0) \\
  u_{j}(1) \\
  \vdots \\
  u_{j}(T-1)
\end{bmatrix}\right)\\
&~~~+\Upsilon_{i,j}\left(\begin{bmatrix}
  w_{i}(0) \\
  w_{i}(1) \\
  \vdots \\
  w_{i}(T-1)
\end{bmatrix}
-\begin{bmatrix}
  w_{j}(0) \\
  w_{j}(1) \\
  \vdots \\
  w_{j}(T-1)
\end{bmatrix}\right)\\
&~~~+\begin{bmatrix}
  \vartheta_{i,j,0} \\
  \vartheta_{i,j,1} \\
  \vdots \\
  \vartheta_{i,j,T-1}
\end{bmatrix}\left(\delta_{i}-\delta_{j}\right),\quad\forall i,j\in\Z_{+}
\endaligned
\end{equation}

\noindent where $\Theta_{i,j}$ is the same as defined in (\ref{eq09}), $\Upsilon_{i,j}$ is some bounded lower triangular matrix in the form of
\begin{equation*}
\aligned
\Upsilon_{i,j}=\begin{bmatrix}
                1 & 0 &\cdots & 0 \\
                \upsilon_{i,j,1}(0) & 1 & \ddots & \vdots \\
                \vdots & \ddots & \ddots & 0\\
                \upsilon_{i,j,T-1}(0) & \cdots &\upsilon_{i,j,T-1}(T-2) & 1
              \end{bmatrix}
\endaligned
\end{equation*}

\noindent and $\vartheta_{i,j,t}$, $\forall t\in\Z_{T-1}$ is some bounded parameter. Further, both (\ref{eq8}) and (\ref{eq04}) hold, and for the same bound $\beta_{\theta}$ as determined in (\ref{eq8}), it simultaneously follows
\begin{equation}\label{eq037}
\aligned
\left|\upsilon_{i,j,t}(\xi)\right|
&\leq\beta_{\theta}, \quad\forall\xi\in\Z_{t-1},\forall t\in\Z_{T-1},\forall i,j\in\Z_{+}\\
\left|\vartheta_{i,j,t}\right|
&\leq\beta_{\theta}, \quad\forall t\in\Z_{T-1},\forall i,j\in\Z_{+}.
\endaligned
\end{equation}
\end{lem}

\begin{proof}
This lemma can be obtained by taking the nonrepetitive uncertainties into account and following the similar way as the proof of Lemma \ref{lem1}. See Appendix \ref{apdx7} for the proof details.
\end{proof}

From Lemma \ref{lem11}, it can be clearly found that the nonrepetitive uncertainties play an important role in influencing the dynamic evolution of ILC along the iteration axis. A specific application of this lemma is to reveal the input-output relation between two sequential iterations, for which (\ref{eq10}) correspondingly becomes
\begin{equation}\label{eq038}
\aligned
\Delta y_{k}(t+1)
&=\sum_{i=0}^{t}\theta_{k,k-1,t}(i)\Delta u_{k}(i)
+\Delta w_{k}(t)+\sum_{i=0}^{t-1}\upsilon_{k,k-1,t}(i)\Delta w_{k}(i)\\
&~~~+\vartheta_{k,k-1,t}\Delta\delta_{k},\quad\forall t\in\Z_{T-1},\forall k\in\Z.
\endaligned
\end{equation}

\noindent Though the effects of the nonrepetitive uncertainties need to be considered, it can be seen from (\ref{eq036}) and (\ref{eq038}) that they play the role as additional inputs during the ILC process. Furthermore, such additional effects can be guaranteed to be bounded under the Assumption (A2). With these observations, we can further generalize the proposed results for optimization-based adaptive ILC to possess certain robustness with respect to nonrepetitive uncertainties, which is shown in the following theorem.

\begin{thm}\label{thm3}
Consider the nonlinear system (\ref{eq029}) satisfying the Assumptions (A1) and (A2). If the Algorithm 1 is applied under the condition (\ref{eq19}), then the following results can be developed for optimization-based adaptive ILC:
\begin{enumerate}
\item
the parameter estimation $\th_{k,k-1,t}(i)$, $\forall i\in\Z_{t}$, $\forall t\in\Z_{T-1}$, $\forall k\in\Z$ is bounded such that (\ref{eq18}) and (\ref{eq018}) hold for some finite bound $\beta_{\th}>0$;

\item
the input $u_{k}(t)$, $\forall t\in\Z_{T-1}$, $\forall k\in\Z_{+}$ and the output $y_{k}(t)$, $\forall t\in\Z_{T}$, $\forall k\in\Z_{+}$ are bounded such that (\ref{eq039}) holds for some finite bounds $\beta_{u}>0$ and $\beta_{y}>0$;

\item
the robust tracking objective (\ref{eq033}) of ILC can be realized; and further, the perfect tracking objective (\ref{eq07}) of ILC can be achieved, provided that (\ref{eq035}) is additionally ensured.
\end{enumerate}
\end{thm}

\begin{proof}
{\it 1):} Note that Lemma \ref{lem2}, and consequently the condition (\ref{eq040}), still hold. With Lemma \ref{lem11}, we can exploit (\ref{eq038}) to obtain
\begin{equation*}\label{}
\aligned
\left|\Delta y_{k-1}(t+1)\right|
&\leq\sqrt{T}\beta_{\theta}\left\|\Delta\overrightarrow{u_{k-1}}(t)\right\|_{2}
+2\beta_{w}+2T\beta_{\theta}\beta_{w}+2\beta_{\theta}\beta_{\delta}
\endaligned
\end{equation*}

\noindent by which we follow the same lines as (\ref{eq041}) to further derive
\begin{equation}\label{eq042}
\aligned
\left\|\overrightarrow{\hat{\theta}_{k,k-1,t}}(t)\right\|_{2}
&\leq\frac{\Ds\mu_{1}}{\Ds\mu_{1}+\mu_{2}}\left\|\overrightarrow{\hat{\theta}_{k-1,k-2,t}}(t)\right\|_{2}
+\sqrt{T}\beta_{\theta}+\frac{\Ds\beta_{w}+T\beta_{\theta}\beta_{w}+\beta_{\theta}\beta_{\delta}}{\Ds\sqrt{\mu_{1}+\mu_{2}}}
\endaligned
\end{equation}

\noindent where $\beta_{w}=\max_{t\in\Z_{T-1}}\beta_{w}(t)$. Based on (\ref{eq042}), we can thus show the boundedness of the parameter estimation in the same way as the proof of Theorem \ref{thm2}.

{\it2) and 3):} If (\ref{eq17}) is combined with (\ref{eq038}), then the dynamics of the tracking error can be described as
\begin{equation}\label{eq043}
\aligned
e_{k}(t+1)
&=e_{k-1}(t+1)-\Delta y_{k}(t+1)\\
&=\left[1-\frac{\Ds\left(\gamma_{1}^{2}+\gamma_{1}\gamma_{2}\right)\theta_{k,k-1,t}(t)\th_{k,k-1,t}(t)}
{\Ds\lambda+\gamma_{1}^{2}\th^{2}_{k,k-1,t}(t)}\right]e_{k-1}(t+1)\\
&~~~-\sum_{i=3}^{m}\frac{\Ds\gamma_{1}\gamma_{i}\theta_{k,k-1,t}(t)\th_{k,k-1,t}(t)}{\Ds\lambda+\gamma_{1}^{2}\th^{2}_{k,k-1,t}(t)}e_{k-i+1}(t+1)
+\widetilde{\kappa}_{k}(t)
\endaligned
\end{equation}

\noindent where $\widetilde{\kappa}_{k}(t)$, in contrast to $\kappa_{k}(t)$ in (\ref{eq26}), is given by
\begin{equation}\label{eq044}
\aligned
\widetilde{\kappa}_{k}(t)
&=\frac{\Ds\gamma_{1}^{2}\theta_{k,k-1,t}(t)\th_{k,k-1,t}(t)}{\Ds\lambda+\gamma_{1}^{2}\th^{2}_{k,k-1,t}(t)}
\sum_{i=0}^{t-1}\th_{k,k-1,t}(i)\Delta u_{k}(i)\\
&~~~-\sum_{i=0}^{t-1}\theta_{k,k-1,t}(i)\Delta u_{k}(i)
-\Delta w_{k}(t)
-\sum_{i=0}^{t-1}\upsilon_{k,k-1,t}(i)\Delta w_{k}(i)
-\vartheta_{k,k-1,t}\Delta\delta_{k}.
\endaligned
\end{equation}

\noindent With (\ref{eq036}), it can be verified that
\begin{equation*}\label{}
\aligned
y_{k-1}(t+1)
&=\theta_{k-1,0,t}(t)u_{k-1}(t)+y_0(t+1)
+\sum_{i=0}^{t-1}\theta_{k-1,0,t}(i)u_{k-1}(i)-\sum_{i=0}^{t}\theta_{k-1,0,t}(i)u_0(i)\\
&~~~+\left[w_{k-1}(t)-w_{0}(t)\right]
+\sum_{i=0}^{t-1}\upsilon_{k-1,0,t}(i)\big[w_{k-1}(i)
-w_{0}(i)\big]
+\vartheta_{k-1,0,t}\left(\delta_{k-1}-\delta_{0}\right)
\endaligned
\end{equation*}

\noindent and then by inserting this into (\ref{eq20}), the dynamics of the input can be described as
\begin{equation}\label{eq045}
\aligned
u_{k}(t)
&=\left[1-\frac{\Ds\left(\gamma_{1}^{2}+\gamma_{1}\gamma_{2}\right)\th_{k,k-1,t}(t)\theta_{k-1,0,t}(t)}
{\Ds\lambda+\gamma_{1}^{2}\th^{2}_{k,k-1,t}(t)}\right]u_{k-1}(t)
+\widetilde{\psi}_k(t)
\endaligned
\end{equation}

\noindent where $\widetilde{\psi}_{k}(t)$, in comparison with $\psi_{k}(t)$ in (\ref{eq23}), is given by
\begin{equation}\label{eq046}
\aligned
\widetilde{\psi}_{k}(t)
&=\frac{\Ds\gamma_{1}\th_{k,k-1,t}(t)}{\Ds\lambda+\gamma_{1}^{2}\th^{2}_{k,k-1,t}(t)}
\Bigg\{\left(\gamma_{1}+\gamma_{2}\right)\sum_{i=0}^{t}\theta_{k-1,0,t}(i)u_{0}(i)
-\big(\gamma_{1}+\gamma_{2}\big)\sum_{i=0}^{t-1}\theta_{k-1,0,t}(i)u_{k-1}(i)\\
&~~~-\gamma_{1}\sum_{i=0}^{t-1}\th_{k,k-1,t}(i)\Delta u_{k}(i)
+\left(\gamma_{1}+\gamma_{2}\right)e_{0}(t+1)+\sum_{i=3}^{m}\gamma_{i}e_{k-i+1}(t+1)
-\left(\gamma_{1}+\gamma_{2}\right)\big[w_{k-1}(t)\\
&~~~-w_{0}(t)\big]-\left(\gamma_{1}+\gamma_{2}\right)\sum_{i=0}^{t-1}\upsilon_{k-1,0,t}(i)\left[w_{k-1}(i)-w_{0}(i)\right]
-\left(\gamma_{1}+\gamma_{2}\right)\vartheta_{k-1,0,t}\left(\delta_{k-1}-\delta_{0}\right)\Bigg\}.
\endaligned
\end{equation}

\noindent Based on (\ref{eq043}), (\ref{eq044}), (\ref{eq045}), and (\ref{eq046}) and with \cite[Lemma 2]{mm:17}, we can establish the results for boundedness of system trajectories and for robust tracking of ILC by following the same steps as the proof of Theorem \ref{thm1}, which is thus omitted here.
\end{proof}

\begin{rem}\label{rem9}
From Theorem \ref{thm3}, it can be seen that the Algorithm 1 of optimization-based adaptive ILC is not only applicable for addressing unknown nonlinear time-varying dynamics but also effective in overcoming ill effect of nonrepetitive uncertainties. This benefits from the optimization-based design of Algorithm 1 and the used DDA approach for ILC. In addition, it is worth emphasizing that it is generally difficult to obtain robustness of data-driven ILC in the presence of nonrepetitive uncertainties, see, e.g., \cite{ch:07}-\cite{byhq:19}. By contrast, Theorem \ref{thm3} successfully shows the robust analysis of data-driven ILC, in spite of nonrepetitive uncertainties arising from disturbances and initial shifts.
\end{rem}

\section{Simulation Tests}\label{sec6}

To illustrate the effectiveness of the proposed optimization-based adaptive ILC algorithm, let the nonlinear system (\ref{eq01}) be given in a specific form of
\begin{equation*}
\aligned
y_{k}(t+1)&=\sin\left(y_{k}(t)\right)+\cos\left(y_{k}(t-1)\right)+\frac{\Ds t+1}{\Ds t+2}u_{k}(t)
+\cos\left(y_{k}(t)\right)\sin\left(u_{k}(t-1)\right)
\endaligned
\end{equation*}

\noindent where the initial output is set as $y_{0}=1.5$. The perfect tracking task (\ref{eq07}) is considered with the desired reference trajectory as
\[
y_{d}(t)=5\sin\left(\frac{\Ds2\pi t}{\Ds50}\right)+0.8\frac{\Ds t(50-t)}{\Ds300},\quad\forall t\in\Z_{50}.
\]

\noindent To implement the Algorithm 1, we adopt the parameters shown in Table \ref{tab1}, and choose the initial estimated value $\th_{0,-1,t}(i)$ such that $\th_{0,-1,t}(i)=0.9$, $\forall i\in\Z_{t}$, $\forall t\in\Z_{T-1}$.
\begin{table}[htpb]
\caption{Parameters Used in Algorithm 1}\label{tab1}
  \centering
  \begin{tabular}{|c|c|c|c|c|c|c|}
    \hline
    $\lambda$ & $\gamma_{1}$  & $\gamma_{2}$ & $\gamma_{3}$ & $\mu_{1}$ & $\mu_{2}$ & $\epsilon$\\
    \hline
    $ 1 $ & 0.8 & 0.14 & 0.06 & 1& 0.001 & 0.01\\
    \hline
  \end{tabular}
\end{table}

\noindent It can be verified that the selection condition (\ref{eq19}) is satisfied.

\begin{figure}[!t]
  \centering
  \includegraphics[width=3in]{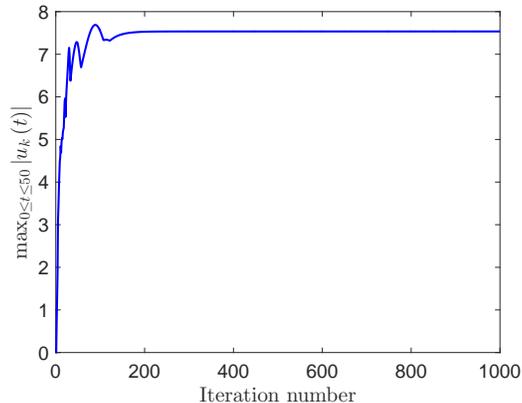}
  \caption{Bounded evolution of the input along the iteration axis.}
  \label{fig1}
\end{figure}

In Fig. \ref{fig1}, the iteration evolution of the input, in the sense of $\max_{t\in\Z_{50}}\left|u_{k}(t)\right|$, is plotted for the first $1000$ iterations. It can be seen from Fig. \ref{fig1} that the input is bounded for all time steps and all iterations. To describe the output tracking performances, we depict the iteration evolution of the tracking error, in the sense of $\max_{t\in\Z_{49}}\left|e_{k}(t+1)\right|$, for the first $1000$ iterations in Fig. \ref{fig2}. It is clear from this figure that the output tracking error converges to zero along the iteration axis. Because the desired reference $y_{d}(t)$ is bounded, Fig. \ref{fig2} implies the boundedness of the output for all time steps and all iterations. In addition, Fig. \ref{fig3} depicts the tracking performance of the system output refined through optimization-based adaptive ILC after $400$ iterations versus the desired reference. It can be obviously revealed from Fig. \ref{fig3} that the perfect output tracking tasks can be achieved for nonlinear systems in spite of unknown nonlinear time-varying dynamics.

\begin{figure}[!t]
  \centering
  \includegraphics[width=3in]{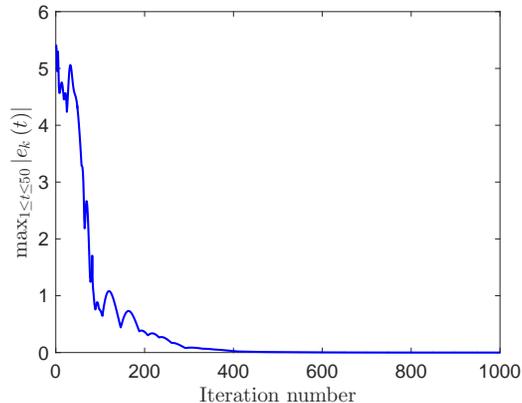}
  \caption{Convergence of the tracking error along the iteration axis.}
  \label{fig2}
\end{figure}
\begin{figure}[!t]
  \centering
  \includegraphics[width=3in]{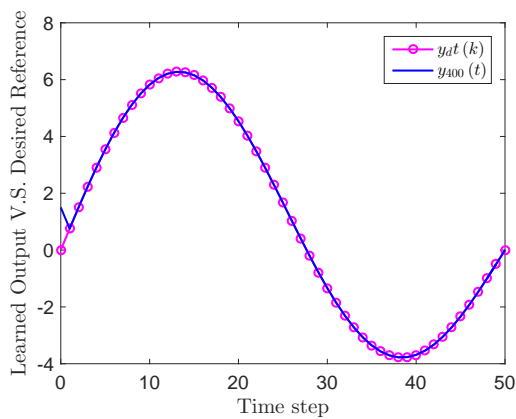}
  \caption{Output tracking performance of optimization-based adaptive ILC after $400$ iterations.}
  \label{fig3}
\end{figure}

Next, we demonstrate the robust performances of the above-considered optimization-based adaptive ILC by instead aiming at the nonlinear system (\ref{eq029}). Without any loss of generality, we adopt the same settings of (\ref{eq01}) for (\ref{eq029}), except the nonrepetitive uncertainties caused by the initial shift $\delta_{k}$ and the disturbance $w_{k}(t)$. Let the nonrepetitive uncertainties be arbitrarily varying with respect to iteration and time, for which we take $\beta_{\delta}=0.01$ and $\beta_{w}=0.01$. Similarly to Figs. \ref{fig1}-\ref{fig3}, Fig. \ref{fig4} depicts the system performances when applying the Algorithm 1 of optimization-based adaptive ILC to (\ref{eq029}). Clearly, we can observe from Fig. \ref{fig4} that the optimization-based adaptive ILC works robustly and effectively, regardless of nonrepetitive uncertainties.

\begin{figure*}[!t]
  \centering
  \includegraphics[width=2.1in]{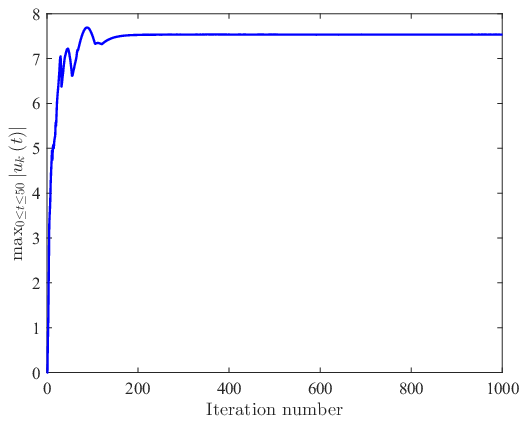}
  \includegraphics[width=2.1in]{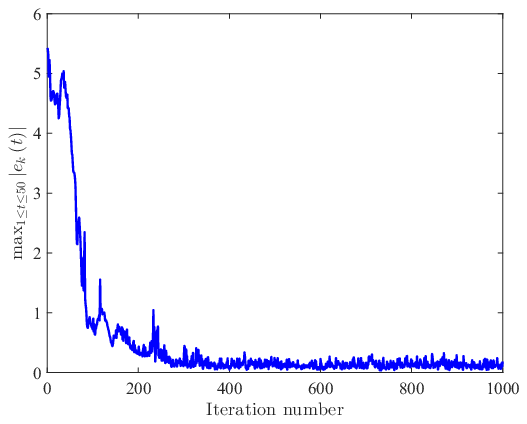}
  \includegraphics[width=2.1in]{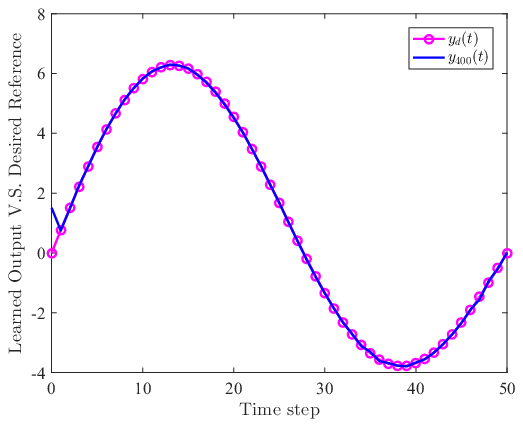}
  \caption{Robust performances of optimization-based adaptive ILC against nonrepetitive uncertainties. Left: boundedness of input. Middle: bounded convergence of the output tracking error. Right: tracking of the desired reference for output learned after $400$ iterations.}
  \label{fig4}
\end{figure*}

{\it Discussions:} The simulation tests in Figs. \ref{fig1}-\ref{fig4} are consistent with our established convergence results of optimization-based adaptive ILC in Theorems \ref{thm2}-\ref{thm3}. This demonstrates the validity of our proposed Algorithm 1, especially the robustness against nonrepetitive uncertainties. In addition, Figs. \ref{fig1}-\ref{fig4} illustrate that our design and analysis of optimization-based adaptive ILC not only generalize the relevant existing results of, e.g., \cite{chhj:15,chjh:18} by proposing a new algorithm but also proceed further to make improvements of them by particularly showing robustness with respect to nonrepetitive uncertainties.

\section{Conclusions}\label{sec7}

In this paper, the algorithm design and convergence analysis of optimization-based adaptive ILC for nonlinear systems have been discussed in spite of unknown time-varying uncertainties. A new design approach has been given to exploit optimization-based adaptive ILC, especially through presenting an improved optimization index to obtain an updating law for the parameter estimation. Simultaneously, a new analysis approach has been proposed to cope with convergence problems for optimization-based adaptive ILC, which resorts to the DDA-based approach to ILC convergence and takes advantage of the good properties of nonnegative matrices. It has been shown that our established results may proceed further with the data-driven ILC problems investigated in, e.g., \cite{ch:07}-\cite{hchh:19}. In addition, robust convergence problems of optimization-based adaptive ILC have been solved in the presence of nonrepetitive uncertainties, despite nonlinear systems subject to iteration-dependent disturbances and initial shifts. Simulation results have also been offered to demonstrate the effectiveness of our obtained results for optimization-based adaptive ILC.

\section*{Acknowledgements}

This work was supported in part by the National Natural Science Foundation of China under grant 61873013 and in part by the Fundamental Research Funds for the Central Universities under grant YWF-19-BJ-J-42.


\begin{appendices}

\section{Proof of Lemma \ref{lem1}}\label{apdx1}

We carry out an inductive analysis on $t$ to prove this lemma, and separate the proof into two steps as follows.

{\it Step a): Let $t=0$.} From (\ref{eq01}), we can obtain
\begin{equation*}
\aligned
y_{k}(1)&=f\left(y_{0},0,\cdots,0,u_{k}(0),0,\cdots,0\right)
\triangleq g^{0}\left(y_{0},u_{k}(0)\right)
\endaligned
\end{equation*}

\noindent with which we can verify
\begin{equation*}
\aligned
\frac{\Ds\partial g^{0}}{\Ds\partial y_{0}}
&=\left.\frac{\Ds\partial f}{\Ds\partial x_{1}}\right|_{\left(y_{0},0,\cdots,0,u_{k}(0),0,\cdots,0\right)},\quad
\frac{\Ds\partial g^{0}}{\Ds\partial u_{k}(0)}
=\left.\frac{\partial f}{\partial x_{l+2}}\right|_{\left(y_{0},0,\cdots,0,u_{k}(0),0,\cdots,0\right)}.
\endaligned
\end{equation*}

\noindent Based on (\ref{eq02}) and (\ref{eq03}), we can thus obtain 
\begin{equation*}
\aligned
\left|\frac{\Ds\partial g^{0}}{\Ds\partial y_{0}}\right|
\leq\beta_{\overline{f}}\triangleq\beta_{\theta}(0),\quad
\frac{\Ds\partial g^{0}}{\Ds\partial u_k(0)}\in\left[\beta_{\underline{f}},\beta_{\overline{f}}\right].
\endaligned
\end{equation*}

{\it Step b): Let any $N\in\Z$ be given.} For $t=0$, $1$, $\cdots$, $N-1$, we proceed with the analysis of step a) to make a hypothesis that $y_{k}(t+1)=g^{t}\left(y_{0},u_{k}(0),\cdots,u_{k}(t)\right)$ holds, and simultaneously,
\begin{equation*}
\aligned
\left|\frac{\Ds\partial g^{t}}{\Ds\partial y_{0}}\right|
&\leq\beta_{\theta}(t),
\left|\frac{\Ds\partial g^{t}}{\Ds\partial u_{k}(0)}\right|
\leq\beta_{\theta}(t),\cdots,
\left|\frac{\Ds\partial g^{t}}{\Ds\partial u_{k}(t-1)}\right|\leq\beta_{\theta}(t),
\frac{\Ds\partial g^{t}}{\Ds\partial u_{k}(t)}
\in\left[\beta_{\underline{f}},\beta_{\overline{f}}\right]
\endaligned
\end{equation*}

\noindent hold for some finite bound $\beta_{\theta}(t)>0$. Then we will prove that the same results can also be developed for $t=N$.

Let $t=N$, and then from (\ref{eq01}), we can exploit the hypothesis made for $t=0$, $1$, $\cdots$, $N-1$ to deduce
\begin{equation*}
\aligned
y_{k}(N+1)
&=f\left(y_{k}(N),\cdots,y_{k}(N-l),u_{k}(N),\cdots,u_{k}(N-n),N\right)\\
&=f\Big(g^{N-1}\left(y_{0},u_{k}(0),\cdots,u_{k}(N-1)\right),\cdots,
g^{N-1-l}\left(y_{0},u_{k}(0),\cdots,u_{k}(N-1-l)\right),\\
&~~~~~~~u_{k}(N),\cdots,u_{k}(N-n),N\Big)\\
& \triangleq g^{N}\left(y_{0},u_{k}(0),\cdots,u_{k}(N)\right).
\endaligned
\end{equation*}

\noindent By following the derivation rules for the compound functions, we consider $g^{N}$ and can obtain
\begin{equation*}
\aligned
\frac{\Ds\partial g^{N}}{\Ds\partial y_{0}}
&=\sum_{i=0}^{l}\frac{\Ds\partial f}{\Ds\partial g^{N-1-i}}\frac{\Ds\partial g^{N-1-i}}{\Ds\partial y_{0}}\\
\frac{\Ds\partial g^{N}}{\Ds\partial u_{k}(0)}
&=\sum_{i=0}^{l}\frac{\Ds\partial f}{\Ds\partial g^{N-1-i}}\frac{\Ds\partial g^{N-1-i}}{\Ds\partial u_{k}(0)}\\
&~\vdots\\
\frac{\Ds\partial g^{N}}{\Ds\partial u_{k}(N-1)}
&=\frac{\Ds\partial f}{\Ds\partial g^{N-1}}\frac{\Ds\partial g^{N-1}}{\Ds\partial u_{k}(N-1)}
+\frac{\Ds\partial f}{\Ds\partial u_{k}(N-1)}\\
\frac{\Ds\partial g^{N}}{\Ds\partial u_{k}(N)}
&=\frac{\Ds\partial f}{\Ds\partial u_{k}(N)}.
\endaligned
\end{equation*}

\noindent This, together with (\ref{eq02}), (\ref{eq06}) and the made hypothesis, leads to
\begin{equation*}
\aligned
\left|\frac{\Ds\partial g^{N}}{\Ds\partial y_{0}}\right|
&\leq\sum_{i=0}^{l}\left|\frac{\Ds\partial f}{\Ds\partial g^{N-1-i}}\right|
\left|\frac{\Ds\partial g^{N-1-i}}{\Ds\partial y_{0}}\right|
\leq\beta_{\theta}(N)\\
\left|\frac{\Ds\partial g^{N}}{\Ds\partial u_{k}(0)}\right|
&\leq\sum_{i=0}^{l}\left|\frac{\Ds\partial f}{\Ds\partial g^{N-1-i}}\right|\left|\frac{\Ds\partial g^{N-1-i}}{\Ds\partial u_{k}(0)}\right|
\leq\beta_{\theta}(N)\\
&~\vdots\\
\left|\frac{\Ds\partial g^{N}}{\Ds\partial u_{k}(N-1)}\right|
&\leq\left|\frac{\Ds\partial f}{\Ds\partial g^{N-1}}\right|\left|\frac{\Ds\partial g^{N-1}}{\Ds\partial u_{k}(N-1)}\right|
+\left|\frac{\Ds\partial f}{\Ds\partial u_{k}(N-1)}\right|
\leq\beta_{\theta}(N)\\
\frac{\Ds\partial g^{N}}{\Ds\partial u_{k}(N)}
&=\frac{\Ds\partial f}{\Ds\partial u_{k}(N)}\in\left[\beta_{\underline{f}},\beta_{\overline{f}}\right]
\endaligned
\end{equation*}

\noindent where we can take $\beta_{\theta}(N)=(l+1)\beta_{\overline{f}}\max_{t\in\Z_{N-1}}\beta_{\theta}(t)+\beta_{\overline{f}}$.

By induction with the above analysis of steps a) and b), we can conclude that for any $t\in\Z_{T-1}$ and $k\in\Z_{+}$,
\begin{equation*}
\aligned
&~~~~~~~~~~~~~y_{k}(t+1)
=g^{t}\left(y_{0},u_{k}(0),\cdots,u_{k}(t)\right)~~\hbox{with}\\
&\left\{
\aligned\left|\frac{\Ds\partial g^{t}}{\Ds\partial y_{0}}\right|
&\leq\beta_{\theta}(t),
\left|\frac{\Ds\partial g^{t}}{\Ds\partial u_{k}(0)}\right|
\leq\beta_{\theta}(t),\cdots,
\left|\frac{\Ds\partial g^{t}}{\Ds\partial u_{k}(t-1)}\right|
\leq\beta_{\theta}(t)\\
\frac{\Ds\partial g^{t}}{\Ds\partial u_{k}(t)}
&\in\left[\beta_{\underline{f}},\beta_{\overline{f}}\right]
\endaligned\right.
\endaligned
\end{equation*}

\noindent where
\[g^{t}:\underbrace{\mathbb{R}\times\mathbb{R}\times\cdots\times\mathbb{R}}_{t+2}\to\mathbb{R}\]

\noindent is some continuously differentiable function, and $\beta_{\theta}(t)>0$ is some finite bound. Let us write $g^{t}$ as $g^{t}\left(z_{1},z_{2},\cdots,z_{t+2}\right)$, where $z_{i}\in\mathbb{R}$, $i=1$, $2$, $\cdots$, $t+2$ denotes the $i$th independent variable of $g^{t}$. Then by employing the differential mean value theorem, we can further obtain
\begin{equation}\label{eq030}
\aligned
y_{i}(t+1)-y_{j}(t+1)
&=\left.\left[\frac{\Ds\partial g^{t}}{\Ds\partial z_{1}},\frac{\Ds\partial g^{t}}{\Ds\partial z_{2}},\cdots,\frac{\Ds\partial g^{t}}{\Ds\partial z_{t+2}}\right]\right|_{\left(z_{1},z_{2},\cdots,z_{t+2}\right)=\left(z_{1}^{\ast},z_{2}^{\ast},\cdots,z_{t+2}^{\ast}\right)}\\
&~~~\times\left(\begin{bmatrix}y_{0}\\u_{i}(0)\\\vdots\\u_{i}(t)\end{bmatrix}
-\begin{bmatrix}y_{0}\\u_{j}(0)\\\vdots\\u_{j}(t)\end{bmatrix}\right)\\
&=\left.\left[\frac{\Ds\partial g^{t}}{\Ds\partial z_{2}},\frac{\Ds\partial g^{t}}{\Ds\partial z_{3}},\cdots,\frac{\Ds\partial g^{t}}{\Ds\partial z_{t+2}}\right]\right|_{\left(z_{1},z_{2},\cdots,z_{t+2}\right)=\left(z_{1}^{\ast},z_{2}^{\ast},\cdots,z_{t+2}^{\ast}\right)}\\
&~~~\times\left(\begin{bmatrix}u_{i}(0)\\u_{i}(1)\\\vdots\\u_{i}(t)\end{bmatrix}
-\begin{bmatrix}u_{j}(0)\\u_{j}(1)\\\vdots\\u_{j}(t)\end{bmatrix}\right)
\endaligned
\end{equation}

\noindent where
\[\aligned
\left(z_{1}^{\ast},z_{2}^{\ast},\cdots,z_{t+2}^{\ast}\right)
&=\varpi\left(y_{0},u_{i}(0),\cdots,u_{i}(t)\right)
+(1-\varpi)\left(y_{0},u_{j}(0),\cdots,u_{j}(t)\right)
\endaligned\]

\noindent for some $\varpi\in[0,1]$. Clearly, we can rewrite (\ref{eq030}) in the compact form of (\ref{eq09}). Let $\beta_{\theta}=\max_{t\in\Z_{T-1}}\beta_{\theta}(t)$, and we can also deduce the boundedness results of $\Theta_{i,j}$ in (\ref{eq8}) and (\ref{eq04}).

\section{Proof of Lemma \ref{lem7}}\label{apdx3}

By noting (\ref{eq11}), we first reformulate (\ref{eq08}) as
\begin{equation}\label{eq12}
\aligned
J\left(u_{k}(t)\right)
&=\lambda\left[\Delta u_{k}(t)\right]^{2}
+\Bigg\{\gamma_{1}\Bigg[-\theta_{k,k-1,t}(t)u_{k}(t)
+\theta_{k,k-1,t}(t)u_{k-1}(t)+e_{k-1}(t+1)\\
&~~~-\sum_{i=0}^{t-1}\theta_{k,k-1,t}(i)\Delta u_{k}(i)\Bigg]
+\sum_{i=2}^{m}\gamma_i e_{k-i+1}(t+1)\Bigg\}^{2}.
\endaligned
\end{equation}

\noindent Then, to determine $u_{k}(t)$ that can optimize (\ref{eq08}), it may generally resort to the condition (see also \cite{chhj:15}-\cite{chhj:18})
\begin{equation}\label{eq05}
\frac{\Ds\partial J(u_k(t))}{\Ds\partial u_k(t)}=0
\end{equation}

\noindent for which we can benefit from (\ref{eq12}) to deduce
\begin{equation}\label{eq13}
\aligned
\frac{\Ds\partial J(u_k(t))}{\Ds\partial u_k(t)}
&=2\lambda\left[u_{k}(t)-u_{k-1}(t))\right]
-2\gamma_1 \theta_{k,k-1,t}(t)\Bigg\{\gamma_{1}\bigg[-\theta_{k,k-1,t}(t)u_{k}(t)\\
&~~~+\theta_{k,k-1,t}(t)u_{k-1}(t)+e_{k-1}(t+1)
-\sum_{i=0}^{t-1}\theta_{k,k-1,t}(i)\Delta u_{k}(i)\bigg]
+\sum_{i=2}^{m}\gamma_i e_{k-i+1}(t+1)\Bigg\}.
\endaligned
\end{equation}

\noindent A straightforward consequence of inserting (\ref{eq13}) into (\ref{eq05}) is an optimal ILC law gained for $t\in\Z_{T-1}$ and $k\in\Z$ in the updating form of (\ref{eq14}).

\section{Proof of Lemma \ref{lem8}}\label{apdx4}

From (\ref{eq011}), we can derive
\begin{equation*}
\aligned
\frac{\Ds\partial H\left(\overrightarrow{\hat{\theta}_{k,k-1,t}}(t)\right)}{\Ds\partial\overrightarrow{\hat{\theta}_{k,k-1,t}}(t)}
&=-2\left[\Delta y_{k-1}(t+1)-\Delta\overrightarrow{u_{k-1}}^{\tp}(t)\overrightarrow{\hat{\theta}_{k,k-1,t}}(t)\right]
\Delta\overrightarrow{u_{k-1}}(t)+2\mu_{1}\bigg[\overrightarrow{\hat{\theta}_{k,k-1,t}}(t)\\
&~~~-\overrightarrow{\hat{\theta}_{k-1,k-2,t}}(t)\bigg]+2\mu_{2}\overrightarrow{\hat{\theta}_{k,k-1,t}}(t).
\endaligned
\end{equation*}

\noindent We take ${\Ds\partial H\left(\overrightarrow{\hat{\theta}_{k,k-1,t}}(t)\right)}\left/{\Ds\partial\overrightarrow{\hat{\theta}_{k,k-1,t}}(t)}\right.=0$ and consequently, can deduce
\begin{equation}\label{eq013}
\aligned
&\left[\mu_{1}I+\mu_{2}I+\Delta\overrightarrow{u_{k-1}}(t)\Delta\overrightarrow{u_{k-1}}^{\tp}(t)\right]\overrightarrow{\hat{\theta}_{k,k-1,t}}(t)
=\Delta y_{k-1}(t+1)\Delta\overrightarrow{u_{k-1}}(t)+\mu_{1}\overrightarrow{\hat{\theta}_{k-1,k-2,t}}(t).
\endaligned
\end{equation}

\noindent In addition, we can verify
\begin{equation*}
\aligned
&\left[\mu_{1}I+\mu_{2}I+\Delta\overrightarrow{u_{k-1}}(t)\Delta\overrightarrow{u_{k-1}}^{\tp}(t)\right]^{-1}
=\frac{\Ds1}{\Ds\mu_1+\mu_2}\left[I-\frac{\Ds\Delta\overrightarrow{u_{k-1}}(t)\Delta\overrightarrow{u_{k-1}}^{\tp}(t)}
{\Ds\mu_{1}+\mu_{2}+\left\|\Delta\overrightarrow{u_{k-1}}(t)\right\|_{2}^{2}}\right]
\endaligned
\end{equation*}

\noindent which, together with (\ref{eq013}), leads to
\begin{equation}\label{eq014}
\aligned
\overrightarrow{\hat{\theta}_{k,k-1,t}}(t)
&=\frac{\Ds\Delta y_{k-1}(t+1)}{\Ds\mu_1+\mu_2}\left[I-\frac{\Ds\Delta\overrightarrow{u_{k-1}}(t)\Delta\overrightarrow{u_{k-1}}^{\tp}(t)}
{\Ds\mu_{1}+\mu_{2}+\left\|\Delta\overrightarrow{u_{k-1}}(t)\right\|_{2}^{2}}\right]\Delta\overrightarrow{u_{k-1}}(t)\\
&~~~+\frac{\Ds\mu_{1}}{\Ds\mu_1+\mu_2}\left[I-\frac{\Ds\Delta\overrightarrow{u_{k-1}}(t)\Delta\overrightarrow{u_{k-1}}^{\tp}(t)}
{\Ds\mu_{1}+\mu_{2}+\left\|\Delta\overrightarrow{u_{k-1}}(t)\right\|_{2}^{2}}\right]\overrightarrow{\hat{\theta}_{k-1,k-2,t}}(t).
\endaligned
\end{equation}

\noindent By noting that
\[\aligned
&\left[I-\frac{\Ds\Delta\overrightarrow{u_{k-1}}(t)\Delta\overrightarrow{u_{k-1}}^{\tp}(t)}
{\Ds\mu_{1}+\mu_{2}+\left\|\Delta\overrightarrow{u_{k-1}}(t)\right\|_{2}^{2}}\right]\Delta\overrightarrow{u_{k-1}}(t)
=\frac{\Ds\mu_{1}+\mu_{2}}{\Ds\mu_{1}+\mu_{2}+\left\|\Delta\overrightarrow{u_{k-1}}(t)\right\|_{2}^{2}}\Delta\overrightarrow{u_{k-1}}(t)
\endaligned\]

\noindent and that
\[\aligned
&\frac{\Ds\Delta\overrightarrow{u_{k-1}}(t)\Delta\overrightarrow{u_{k-1}}^{\tp}(t)}
{\Ds\mu_{1}+\mu_{2}+\left\|\Delta\overrightarrow{u_{k-1}}(t)\right\|_{2}^{2}}\overrightarrow{\hat{\theta}_{k-1,k-2,t}}(t)
=\frac{\Ds\Delta\overrightarrow{u_{k-1}}^{\tp}(t)\overrightarrow{\hat{\theta}_{k-1,k-2,t}}(t)}
{\Ds\mu_{1}+\mu_{2}+\left\|\Delta\overrightarrow{u_{k-1}}(t)\right\|_{2}^{2}}\Delta\overrightarrow{u_{k-1}}(t)
\endaligned\]

\noindent we can employ (\ref{eq014}) to further obtain
\begin{equation*}\label{}
\aligned
\overrightarrow{\hat{\theta}_{k,k-1,t}}(t)
&=\frac{\Ds\Delta y_{k-1}(t+1)}{\Ds\mu_{1}+\mu_{2}+\left\|\Delta\overrightarrow{u_{k-1}}(t)\right\|_{2}^{2}}\Delta\overrightarrow{u_{k-1}}(t)
+\frac{\Ds\mu_{1}}{\Ds\mu_1+\mu_2}\overrightarrow{\hat{\theta}_{k-1,k-2,t}}(t)\\
&~~~-\frac{\Ds\mu_{1}}{\Ds\mu_1+\mu_2}\frac{\Ds\Delta\overrightarrow{u_{k-1}}(t)\Delta\overrightarrow{u_{k-1}}^{\tp}(t)}
{\Ds\mu_{1}+\mu_{2}+\left\|\Delta\overrightarrow{u_{k-1}}(t)\right\|_{2}^{2}}\overrightarrow{\hat{\theta}_{k-1,k-2,t}}(t)\\
&=\frac{\Ds\mu_{1}}{\Ds\mu_{1}+\mu_{2}}\overrightarrow{\hat{\theta}_{k-1,k-2,t}}(t)
+\frac{\Ds1}{\Ds\mu_{1}+\mu_{2}+\left\|\Delta\overrightarrow{{u}_{k-1}}(t)\right\|_{2}^{2}}\bigg[\Delta y_{k-1}(t+1)\\
&~~~-\frac{\Ds\mu_1}{\Ds\mu_{1}+\mu_{2}}\Delta\overrightarrow{{u}_{k-1}}^{\tp}(t)\overrightarrow{\hat{\theta}_{k-1,k-2,t}}(t)\bigg]
\Delta\overrightarrow{{u}_{k-1}}(t)
\endaligned
\end{equation*}

\noindent namely, (\ref{eq012}) holds. 

\section{Proof of Lemma \ref{lem2}}\label{apdx6}

We can easily see from (\ref{eq015}) that $Q\left(\Delta\overrightarrow{u_{k-1}}(t)\right)\in\mathbb{R}^{(t+1)\times(t+1)}$ is a real symmetric matrix, and can thus obtain
\begin{equation}\label{eq016}
\left\|Q\left(\Delta\overrightarrow{u_{k-1}}(t)\right)\right\|_{2}
=\rho\left(Q\left(\Delta\overrightarrow{u_{k-1}}(t)\right)\right).
\end{equation}

\noindent We can also verify that for $\Delta\overrightarrow{u_{k-1}}(t)\Delta\overrightarrow{u_{k-1}}^{\tp}(t)$, the eigenvalues are either $\left\|\Delta\overrightarrow{u_{k-1}}(t)\right\|_{2}^{2}$ or $0$ (with a multiplicity of $t$). By noting this fact, we can further get from (\ref{eq015}) that for $Q\left(\Delta\overrightarrow{u_{k-1}}(t)\right)$, the eigenvalues are either $\left(\mu_{1}+\mu_{2}\right)\left/\left(\mu_{1}+\mu_{2}+\left\|\Delta\overrightarrow{u_{k-1}}(t)\right\|_{2}^{2}\right)\right.$ or $1$ (with a multiplicity of $t$). As a direct consequence, we have
\begin{equation*}
\aligned
\rho\left(Q\left(\Delta\overrightarrow{u_{k-1}}(t)\right)\right)
=\left\{
\aligned
&1, &t&>0\\
&\frac{\Ds\mu_{1}+\mu_{2}}{\Ds\mu_{1}+\mu_{2}+\left\|\Delta\overrightarrow{u_{k-1}}(t)\right\|_{2}^{2}}, &t&=0
\endaligned,~\forall k\geq2
\right.
\endaligned
\end{equation*}

\noindent which, together with (\ref{eq016}), yields
\begin{equation*}
\left\|Q\left(\Delta\overrightarrow{u_{k-1}}(t)\right)\right\|_{2}
=\rho\left(Q\left(\Delta\overrightarrow{u_{k-1}}(t)\right)\right)\leq1.
\end{equation*}

\noindent Namely, Lemma \ref{lem2} is developed.

\section{Proof of Lemma \ref{lem4}}\label{apdx5}

To prove Lemma \ref{lem4}, we exploit the properties of nonnegative matrices to disclose the relationship between the condition (\ref{eq27}) in Lemma \ref{lem4} and the condition (C) in Lemma \ref{lem10}.

\begin{lem}\label{lem5}
For any $t\in\Z_{T-1}$, the condition (C) in Lemma \ref{lem10} holds as a consequence of the condition (\ref{eq27}) in Lemma \ref{lem4}.
\end{lem}

\begin{proof}
Let us take $\omega_{s}(t)=(m-1)s+1$, $\forall t\in\Z_{T-1}$ in (\ref{eq023}), and we can use the properties of nonnegative matrices to deduce
\begin{equation}\label{eq31}
\aligned
\left\|\prod_{k=(m-1)s+1}^{(m-1)s+m-1}P_{k}(t)\right\|_{\infty}
& \leq\left\|\prod_{k=(m-1)s+1}^{(m-1)s+m-1}\left|P_{k}(t)\right|\right\|_{\infty}\\
&=\left\|\prod_{k=(m-1)s+1}^{(m-1)s+m-1}\left|P_{k}(t)\right|\bm{1}_{m-1}\right\|_{\infty}.
\endaligned
\end{equation}

\noindent To proceed with (\ref{eq31}), we adopt an inductive analysis approach to show a property that if (\ref{eq27}) holds, then
\begin{equation}\label{eq32}
\prod_{k=(m-1)s+1}^{(m-1)s+m-1}\left|P_{k}(t)\right|\bm{1}_{m-1}\leq\zeta\bm{1}_{m-1}.
\end{equation}

{\it Step 1):} For $i=1$, we consider
\[
\prod_{k=(m-1)s+1}^{(m-1)s+i}\left|P_{k}(t)\right|\bm{1}_{m-1}=\left|P_{(m-1)s+1}(t)\right|\bm{1}_{m-1}
\]

\noindent and then by employing the definition of the nonnegative matrix $\left|P_{(m-1)s+1}(t)\right|$, we can gain that its induced nonnegative vector $\left|P_{(m-1)s+1}(t)\right|\bm{1}_{m-1}$ satisfies
\begin{equation}\label{eq33}
\aligned
\left|P_{(m-1)s+1}(t)\right|\bm{1}_{m-1}
&=\begin{bmatrix}
\aligned\sum_{j=1}^{m-1}\left|p_{j,(m-1)s+1}(t)\right|\endaligned\\
\bm{1}_{m-2}
\end{bmatrix}
\leq
\begin{bmatrix}
\zeta \\
\bm{1}_{m-2}
\end{bmatrix}
\triangleq
\begin{bmatrix}
\zeta\bm{1}_{i} \\
\bm{1}_{m-i-1}
\end{bmatrix}.
\endaligned
\end{equation}

{\it Step 2):} For any $i\geq1$, we explore the fact (\ref{eq33}) to make the following hypothesis:
\begin{equation}\label{eq024}
\aligned
\prod_{k=(m-1)s+1}^{(m-1)s+i}\left|P_{k}(t)\right|\bm{1}_{m-1}
\leq
\begin{bmatrix}
\zeta\bm{1}_{i} \\
\bm{1}_{m-i-1}
\end{bmatrix}.
\endaligned
\end{equation}

\noindent Then for the next step $i+1$, we insert (\ref{eq024}) and can again apply the properties of nonnegative matrices to derive
\[\aligned
\prod_{k=(m-1)s+1}^{(m-1)s+i+1}\left|P_{k}(t)\right|\bm{1}_{m-1}
&=\left|P_{(m-1)s+i+1}(t)\right|\left[\prod_{l=(m-1)s+1}^{(m-1)s+i}\left|P_{k}(t)\right|\bm{1}_{m-1}\right]\\
&\leq\left|P_{(m-1)s+i+1}(t)\right|
\begin{bmatrix}
\zeta\bm{1}_{i} \\
\bm{1}_{m-i-1}
\end{bmatrix}
\endaligned\]

\noindent which, together with the definition of $\left|P_{k}(t)\right|$ and the condition (\ref{eq27}), leads to
\begin{equation}\label{eq34}
\aligned
\prod_{k=(m-1)s+1}^{(m-1)s+i+1}\left|P_{k}(t)\right|\bm{1}_{m-1}
&\leq
\begin{bmatrix}
\left(\aligned
&\zeta\sum_{j=1}^{i}\left|p_{j,(m-1)s+i+1}(t)\right|\\
&+\sum_{j=i+1}^{m-1}\left|p_{j,(m-1)s+i+1}(t)\right|
\endaligned\right)\\
\zeta \bm{1}_{i}\\
\bm{1}_{m-i-2}
\end{bmatrix}\\
&\leq
\begin{bmatrix}
\aligned\sum_{j=1}^{m-1}\left|p_{j,(m-1)s+i+1}(t)\right|\endaligned\\
\zeta\bm{1}_{i}\\
\bm{1}_{m-i-2}
\end{bmatrix}\\
&\leq
\begin{bmatrix}
\zeta\bm{1}_{i+1}\\
\bm{1}_{m-i-2}
\end{bmatrix}.
\endaligned
\end{equation}

\noindent Clearly, (\ref{eq34}) implies that the hypothesis made in (\ref{eq024}) can also hold by updating $i$ with $i+1$.

With the above analysis of steps 1) and 2) and by induction, we can conclude that (\ref{eq32}) holds. By combining (\ref{eq31}) and (\ref{eq32}), we can further deduce
\begin{equation}\label{eq14k}
\aligned
\left\|\prod_{k=(m-1)s+1}^{(m-1)s+m-1}P_{k}(t)\right\|_{\infty}
&\leq\left\|\prod_{k=(m-1)s+1}^{(m-1)s+m-1}\left|P_{k}(t)\right|\bm{1}_{m-1}\right\|_{\infty}\\
&\leq\left\|\zeta\bm{1}_{m-1}\right\|_{\infty}\\
&=\zeta<1.
\endaligned
\end{equation}

\noindent Consequently, (\ref{eq14k}) guarantees that the condition (C) in Lemma \ref{lem10} can be developed by particularly setting  $\omega_{s}(t)=(m-1)s+1$, $\forall t\in\Z_{T-1}$ and $\eta=\zeta$.
\end{proof}

In addition, for the relationship between (\ref{eq25}) and (\ref{eq29}), we clearly have two equivalent results in the following lemma.

\begin{lem}\label{lem9}
For any $t\in\Z_{T-1}$, it follows:
\begin{enumerate}
\item
$\overrightarrow{e_{k}}(t+1)$ is bounded (respectively, $\lim_{k\to\infty}\overrightarrow{e_{k}}(t+1)=0$ if and only if
$e_{k}(t+1)$ is bounded (respectively, $\lim_{k\to\infty}e_{k}(t+1)=0$);

\item
$\overrightarrow{\kappa_{k}}(t)$ is bounded (respectively, $\lim_{k\to\infty}\overrightarrow{\kappa_{k}}(t)=0$ if and only if
$\kappa_{k}(t)$ is bounded (respectively, $\lim_{k\to\infty}\kappa_{k}(t)=0$).
\end{enumerate}
\end{lem}

\begin{proof}
A consequence of the definitions for $\overrightarrow{e_{k}}(t+1)$ and $\overrightarrow{\kappa_{k}}(t)$ in (\ref{eq28}).
\end{proof}

With Lemmas \ref{lem5} and \ref{lem9}, we can prove Lemma \ref{lem4} as follows.

\begin{proof}[Proof of Lemma \ref{lem4}]
Based on Lemmas \ref{lem10} and \ref{lem5}, we know that if the condition (\ref{eq27}) holds, then $\lim_{k\to\infty}\overrightarrow{e_{k}}(t+1)=0$, provided that $\lim_{k\to\infty}\overrightarrow{\kappa_{k}}(t)=0$. By the two equivalent results of Lemma \ref{lem9}, we can further conclude $\lim_{k\to\infty}e_{k}(t+1)=0$, provided that $\lim_{k\to\infty}\kappa_{k}(t)=0$. In the same way, we can prove that $e_{k}(t+1)$ is bounded, provided that $\kappa_{k}(t)$ is bounded. Namely, Lemma \ref{lem4} is obtained.
\end{proof}

\section{Proof of Lemma \ref{lem6}}\label{apdx2}

With Lemma \ref{lem1} and Theorem \ref{thm2}, we can obtain
\begin{equation}\label{eq031}
\aligned
\th_{k,k-1,t}^{2}(t)
\leq\beta_{\th}^{2},
\beta_{\underline{f}}\epsilon
\leq\theta_{k,k-1,t}(t)\th_{k,k-1,t}(t),
\frac{\Ds\left(\gamma_{1}^{2}+\gamma_{1}\gamma_{2}\right)\theta_{k,k-1,t}(t)\th_{k,k-1,t}(t)}{\Ds\lambda+\gamma_{1}^{2}\th^{2}_{k,k-1,t}(t)}
\leq\frac{\Ds\left(\gamma_{1}^{2}+\gamma_{1}\gamma_{2}\right)\beta_{\overline{f}}\beta_{\th}}{\Ds\lambda}
\endaligned
\end{equation}

\noindent which, together with (\ref{eq19}), leads to
\begin{equation}\label{eq032}
\aligned
\frac{\Ds\gamma_{1}\left(\gamma_{1}+\gamma_{2}-\sum_{i=3}^{m}\gamma_{i}\right)\beta_{\underline{f}}\epsilon}{\Ds\lambda+\gamma_{1}^{2}\beta_{\th}^{2}}
\leq\frac{\Ds\left(\gamma_{1}^{2}+\gamma_{1}\gamma_{2}\right)\beta_{\overline{f}}\beta_{\th}}{\Ds\lambda}
<1.
\endaligned
\end{equation}

\noindent By inserting (\ref{eq031}) and (\ref{eq032}) into (\ref{eq27}), we can verify
\begin{equation*}\label{}
\aligned
&\left |1-\frac{\Ds\left(\gamma_{1}^{2}+\gamma_{1}\gamma_{2}\right)\theta_{k,k-1,t}(t)\th_{k,k-1,t}(t)}
{\Ds\lambda+\gamma_{1}^{2}\th^{2}_{k,k-1,t}(t)}\right|
+\sum_{i=3}^{m}\left|\frac{\Ds\gamma_{1}\gamma_{i}\theta_{k,k-1,t}(t)\th_{k,k-1,t}(t)}{\Ds\lambda+\gamma_{1}^{2}\th^{2}_{k,k-1,t}(t)}\right|\\
&~~~=1-\frac{\Ds\left(\gamma_{1}^{2}+\gamma_{1}\gamma_{2}\right)\theta_{k,k-1,t}(t)\th_{k,k-1,t}(t)}
{\Ds\lambda+\gamma_{1}^{2}\th^{2}_{k,k-1,t}(t)}
+\sum_{i=3}^{m}\frac{\Ds\gamma_{1}\gamma_{i}\theta_{k,k-1,t}(t)\th_{k,k-1,t}(t)}{\Ds\lambda+\gamma_{1}^{2}\th^{2}_{k,k-1,t}(t)}\\
&~~~=1-\frac{\Ds\gamma_{1}\left(\gamma_{1}+\gamma_{2}-\sum_{i=3}^{m}\gamma_{i}\right)\theta_{k,k-1,t}(t)\th_{k,k-1,t}(t)}
{\Ds\lambda+\gamma_{1}^{2}\th^{2}_{k,k-1,t}(t)}\\
&~~~\leq1-\frac{\Ds\gamma_{1}\left(\gamma_{1}+\gamma_{2}-\sum_{i=3}^{m}\gamma_{i}\right)\beta_{\underline{f}}\epsilon}
{\Ds\lambda+\gamma_{1}^{2}\beta_{\th}^{2}}\\
&~~~\triangleq\zeta<1,\quad\forall t\in\Z_{T-1,}\forall k\in\Z
\endaligned
\end{equation*}

\noindent that is, the condition (\ref{eq27}) holds.

In the same way as (\ref{eq031}), we can apply (\ref{eq04}) and (\ref{eq018}) to derive
\begin{equation*}
\beta_{\underline{f}}\epsilon
\leq\th_{k,k-1,t}(t)\theta_{k-1,0,t}(t)
\leq\beta_{\overline{f}}\beta_{\th}
\end{equation*}

\noindent and further by (\ref{eq19}), we can deduce
\begin{equation*}
\aligned
\frac{\Ds\left(\gamma_{1}^{2}+\gamma_{1}\gamma_{2}\right)\beta_{\underline{f}}\epsilon}{\Ds\lambda+\gamma_{1}^{2}\beta_{\th}^{2}}
&\leq\frac{\Ds\left(\gamma_{1}^{2}+\gamma_{1}\gamma_{2}\right)\th_{k,k-1,t}(t)\theta_{k-1,0,t}(t)}
{\Ds\lambda+\gamma_{1}^{2}\th^{2}_{k,k-1,t}(t)}\\
&\leq
\frac{\Ds\left(\gamma_{1}^{2}+\gamma_{1}\gamma_{2}\right)\beta_{\overline{f}}\beta_{\th}}{\Ds\lambda}\\
&<1
\endaligned
\end{equation*}

\noindent which, together with (\ref{eq24}), results in
\begin{equation*}
\aligned
\left|1-\frac{\Ds\left(\gamma_{1}^{2}+\gamma_{1}\gamma_{2}\right)\th_{k,k-1,t}(t)\theta_{k-1,0,t}(t)}
{\Ds\lambda+\gamma_{1}^{2}\th^{2}_{k,k-1,t}(t)}\right|
&=1-\frac{\Ds\left(\gamma_{1}^{2}+\gamma_{1}\gamma_{2}\right)\th_{k,k-1,t}(t)\theta_{k-1,0,t}(t)}
{\Ds\lambda+\gamma_{1}^{2}\th^{2}_{k,k-1,t}(t)}\\
&\leq1-\frac{\Ds\left(\gamma_{1}^{2}+\gamma_{1}\gamma_{2}\right)\beta_{\underline{f}}\epsilon}{\Ds\lambda+\gamma_{1}^{2}\beta_{\th}^{2}}\\
&\triangleq\phi<1,\quad\forall t\in\Z_{T-1},\forall k\in\Z.
\endaligned
\end{equation*}

\noindent Namely, the condition (\ref{eq24}) holds.

\section{Proof of Lemma \ref{lem11}}\label{apdx7}

Similarly to the proof of Lemma \ref{lem1}, an inductive analysis on $t$ is performed to prove this lemma, and the proof is separated into two steps as follows.

{\it Step a): Let $t=0$.} Then the use of (\ref{eq029}) gives
\begin{equation*}
\aligned
y_{k}(1)&=f\left(y_k(0),0,\cdots,0,u_{k}(0),0,\cdots,0\right)+w_k(0)\\
&\triangleq \overline{g}^{0}\left(y_k(0),u_{k}(0),w_k(0)\right)
\endaligned
\end{equation*}

\noindent based on which we have
\begin{equation*}
\aligned
\frac{\Ds\partial \overline{g}^{0}}{\Ds\partial y_k(0)}
&=\left.\frac{\Ds\partial f}{\Ds\partial x_{1}}\right|_{\left(y_k(0),0,\cdots,0,u_{k}(0),0,\cdots,0\right)}\\
\frac{\Ds\partial \overline{g}^{0}}{\Ds\partial u_{k}(0)}
&=\left.\frac{\partial f}{\partial x_{l+2}}\right|_{\left(y_k(0),0,\cdots,0,u_{k}(0),0,\cdots,0\right)}\\
\frac{\Ds\partial \overline{g}^{0}}{\Ds\partial w_k(0)}
&=1.
\endaligned
\end{equation*}

\noindent By employing (\ref{eq02}) and (\ref{eq03}), we can further derive 
\begin{equation*}
\aligned
\left|\frac{\Ds\partial \overline{g}^{0}}{\Ds\partial y_k(0)}\right|
\leq\beta_{\overline{f}}\triangleq\beta_{\theta}(0)&,\quad
\frac{\Ds\partial \overline{g}^{0}}{\Ds\partial u_k(0)}\in\left[\beta_{\underline{f}},\beta_{\overline{f}}\right],\quad
\frac{\Ds\partial \overline{g}^{0}}{\Ds\partial w_k(0)}=1.
\endaligned
\end{equation*}

{\it Step b): Let us consider any $N\in\Z$.} For $t=0$, $1$, $\cdots$, $N-1$, we assume $y_{k}(t+1)=\overline{g}^{t}\left(y_{0},u_{k}(0),\cdots,\right.$ $\left.u_{k}(t),w_{k}(0),\cdots,w_{k}(t)\right)$ and simultaneously that it satisfies
\begin{equation*}
\aligned
\left|\frac{\Ds\partial \overline{g}^{t}}{\Ds\partial y_k(0)}\right|
&\leq\beta_{\theta}(t),
\frac{\Ds\partial \overline{g}^{t}}{\Ds\partial u_{k}(t)}
\in\left[\beta_{\underline{f}},\beta_{\overline{f}}\right],
\frac{\Ds\partial \overline{g}^{t}}{\Ds\partial w_{k}(t)}=1\\
\left|\frac{\Ds\partial \overline{g}^{t}}{\Ds\partial u_{k}(0)}\right|
&\leq\beta_{\theta}(t),\cdots,
\left|\frac{\Ds\partial \overline{g}^{t}}{\Ds\partial u_{k}(t-1)}\right|\leq\beta_{\theta}(t),
\\
\left|\frac{\Ds\partial \overline{g}^{t}}{\Ds\partial w_{k}(0)}\right|
&\leq\beta_{\theta}(t),\cdots,
\left|\frac{\Ds\partial \overline{g}^{t}}{\Ds\partial w_{k}(t-1)}\right|
\leq\beta_{\theta}(t)
\endaligned
\end{equation*}

\noindent for some finite bound $\beta_{\theta}(t)>0$. Next, we show that for $t=N$, we can deduce the same results.

When we consider (\ref{eq029}) for $t=N$, the use of the hypothesis made for $t=0$, $1$, $\cdots$, $N-1$ leads to
\begin{equation*}
\aligned
y_{k}(N+1)
&=f\left(y_{k}(N),\cdots,y_{k}(N-l),u_{k}(N),\cdots,u_{k}(N-n),N\right)
+w_{k}(N)\\
&=f\Big(\overline{g}^{N-1},\cdots,\overline{g}^{N-1-l},u_{k}(N),\cdots,u_{k}(N-n),N\Big)
+w_{k}(N)\\
&\triangleq\overline{g}^{N}\left(y_k(0),u_{k}(0),\cdots,u_{k}(N),w_{k}(0),\cdots,w_{k}(N)\right).
\endaligned
\end{equation*}

\noindent For $\overline{g}^{N}$, we employ the derivation rules of compound functions to deduce
\begin{equation*}
\aligned
\frac{\Ds\partial \overline{g}^{N}}{\Ds\partial y_k(0)}
&=\sum_{i=0}^{l}\frac{\Ds\partial f}{\Ds\partial \overline{g}^{N-1-i}}\frac{\Ds\partial \overline{g}^{N-1-i}}{\Ds\partial y_k(0)}\\
\frac{\Ds\partial \overline{g}^{N}}{\Ds\partial u_{k}(0)}
&=\sum_{i=0}^{l}\frac{\Ds\partial f}{\Ds\partial \overline{g}^{N-1-i}}\frac{\Ds\partial \overline{g}^{N-1-i}}{\Ds\partial u_{k}(0)}\\
&~\vdots\\
\frac{\Ds\partial \overline{g}^{N}}{\Ds\partial u_{k}(N-1)}
&=\frac{\Ds\partial f}{\Ds\partial \overline{g}^{N-1}}\frac{\Ds\partial \overline{g}^{N-1}}{\Ds\partial u_{k}(N-1)}
+\frac{\Ds\partial f}{\Ds\partial u_{k}(N-1)}\\
\frac{\Ds\partial \overline{g}^{N}}{\Ds\partial u_{k}(N)}
&=\frac{\Ds\partial f}{\Ds\partial u_{k}(N)}
\endaligned
\end{equation*}

\noindent and
\begin{equation*}
\aligned
\frac{\Ds\partial \overline{g}^{N}}{\Ds\partial w_{k}(0)}
&=\sum_{i=0}^{l}\frac{\Ds\partial f}{\Ds\partial \overline{g}^{N-1-i}}\frac{\Ds\partial \overline{g}^{N-1-i}}{\Ds\partial w_{k}(0)}\\
&~\vdots\\
\frac{\Ds\partial \overline{g}^{N}}{\Ds\partial w_{k}(N-1)}
&=\frac{\Ds\partial f}{\Ds\partial \overline{g}^{N-1}}\frac{\Ds\partial \overline{g}^{N-1}}{\Ds\partial w_{k}(N-1)}\\
\frac{\Ds\partial \overline{g}^{N}}{\Ds\partial w_{k}(N)}&=1.
\endaligned
\end{equation*}

\noindent Again with the made hypothesis and by inserting (\ref{eq02}) and (\ref{eq06}), we can obtain
\begin{equation*}
\aligned
\left|\frac{\Ds\partial \overline{g}^{N}}{\Ds\partial y_k(0)}\right|
&\leq\sum_{i=0}^{l}\left|\frac{\Ds\partial f}{\Ds\partial \overline{g}^{N-1-i}}\right|
\left|\frac{\Ds\partial \overline{g}^{N-1-i}}{\Ds\partial y_k(0)}\right|
\leq\beta_{\theta}(N)\\
\left|\frac{\Ds\partial \overline{g}^{N}}{\Ds\partial u_{k}(0)}\right|
&\leq\sum_{i=0}^{l}\left|\frac{\Ds\partial f}{\Ds\partial \overline{g}^{N-1-i}}\right|\left|\frac{\Ds\partial \overline{g}^{N-1-i}}{\Ds\partial u_{k}(0)}\right|
\leq\beta_{\theta}(N)\\
&~\vdots\\
\left|\frac{\Ds\partial \overline{g}^{N}}{\Ds\partial u_{k}(N-1)}\right|
&\leq\left|\frac{\Ds\partial f}{\Ds\partial \overline{g}^{N-1}}\right|\left|\frac{\Ds\partial \overline{g}^{N-1}}{\Ds\partial u_{k}(N-1)}\right|
+\left|\frac{\Ds\partial f}{\Ds\partial u_{k}(N-1)}\right|
\leq\beta_{\theta}(N)\\
\frac{\Ds\partial \overline{g}^{N}}{\Ds\partial u_{k}(N)}
&=\frac{\Ds\partial f}{\Ds\partial u_{k}(N)}\in\left[\beta_{\underline{f}},\beta_{\overline{f}}\right]
\endaligned
\end{equation*}

\noindent and
\begin{equation*}
\aligned
\left|\frac{\Ds\partial \overline{g}^{N}}{\Ds\partial w_{k}(0)}\right|
&\leq\sum_{i=0}^{l}\left|\frac{\Ds\partial f}{\Ds\partial \overline{g}^{N-1-i}}\right|\left|\frac{\Ds\partial \overline{g}^{N-1-i}}{\Ds\partial w_{k}(0)}\right|
\leq\beta_{\theta}(N)\\
&~\vdots\\
\left|\frac{\Ds\partial \overline{g}^{N}}{\Ds\partial w_{k}(N-1)}\right|
&\leq\left|\frac{\Ds\partial f}{\Ds\partial \overline{g}^{N-1}}\right|\left|\frac{\Ds\partial \overline{g}^{N-1}}{\Ds\partial w_{k}(N-1)}\right|\leq \beta_{\theta}(N)\\
\frac{\Ds\partial \overline{g}^{N}}{\Ds\partial w_{k}(N)}
&=1
\endaligned
\end{equation*}

\noindent where $\beta_{\theta}(N)=(l+1)\beta_{\overline{f}}\max_{t\in\Z_{N-1}}\beta_{\theta}(t)+\beta_{\overline{f}}$ can be adopted as a candidate.

Based on the analysis of the above steps a) and b), we can conclude by induction that for any $t\in\Z_{T-1}$ and $k\in\Z_{+}$,
\begin{equation*}
\aligned
&y_{k}(t+1)
=\overline{g}^{t}\left(y_k(0),u_{k}(0),\cdots,u_{k}(t),w_{k}(0),\cdots,w_{k}(t)\right)~~\hbox{with}\\
&~~~~~\left\{
\aligned
\left|\frac{\Ds\partial \overline{g}^{t}}{\Ds\partial y_k(0)}\right|
&\leq\beta_{\theta}(t),
\frac{\Ds\partial \overline{g}^{t}}{\Ds\partial u_{k}(t)}
\in\left[\beta_{\underline{f}},\beta_{\overline{f}}\right],
\frac{\Ds\partial \overline{g}^{t}}{\Ds\partial w_{k}(t)}=1\\
\left|\frac{\Ds\partial \overline{g}^{t}}{\Ds\partial u_{k}(0)}\right|
&\leq\beta_{\theta}(t),\cdots,
\left|\frac{\Ds\partial \overline{g}^{t}}{\Ds\partial u_{k}(t-1)}\right|\leq\beta_{\theta}(t),
\\
\left|\frac{\Ds\partial \overline{g}^{t}}{\Ds\partial w_{k}(0)}\right|
&\leq\beta_{\theta}(t),\cdots,
\left|\frac{\Ds\partial \overline{g}^{t}}{\Ds\partial w_{k}(t-1)}\right|
\leq\beta_{\theta}(t)
\endaligned\right.
\endaligned
\end{equation*}

\noindent where $\overline{g}^{t}$ is some continuously differentiable function such that
\[\overline{g}^{t}:\underbrace{\mathbb{R}\times\mathbb{R}\times\cdots\times\mathbb{R}}_{2t+3}\to\mathbb{R}\]

\noindent and $\beta_{\theta}(t)>0$ is some finite bound. For convenience, we write $\overline{g}^{t}$ in terms of $\overline{g}^{t}\left(z_{1},z_{2},\cdots,z_{2t+3}\right)$, where $z_{i}\in\mathbb{R}$, $i=1$, $2$, $\cdots$, $2t+3$ represents the $i$th independent variable of $\overline{g}^{t}$. Then based on the use of the mean value theorem (see, e.g., \cite[P. 651]{k:02}), we can validate
\begin{equation}\label{eq0300}
\aligned
y_{i}(t+1)-y_{j}(t+1)
&=\left.\left[\frac{\Ds\partial \overline{g}^{t}}{\Ds\partial z_{1}},\frac{\Ds\partial \overline{g}^{t}}{\Ds\partial z_{2}},\cdots,\frac{\Ds\partial \overline{g}^{t}}{\Ds\partial z_{2t+3}}\right]\right|_{\left(z_{1},z_{2},\cdots,z_{2t+3}\right)=\left(z_{1}^{\ast},z_{2}^{\ast},\cdots,z_{2t+3}^{\ast}\right)}\\
&~~~\times\left(\begin{bmatrix}y_i(0)\\u_{i}(0)\\\vdots\\u_{i}(t)\\w_{i}(0)\\\vdots\\w_{i}(t)\end{bmatrix}
-\begin{bmatrix}y_j(0)\\u_{j}(0)\\\vdots\\u_{j}(t)\\w_{j}(0)\\\vdots\\w_{j}(t)\end{bmatrix}\right)\\
&=\left.\left[\frac{\Ds\partial \overline{g}^{t}}{\Ds\partial z_{2}},\frac{\Ds\partial \overline{g}^{t}}{\Ds\partial z_{3}},\cdots,\frac{\Ds\partial \overline{g}^{t}}{\Ds\partial z_{t+2}}\right]\right|_{\left(z_{1},z_{2},\cdots,z_{2t+3}\right)=\left(z_{1}^{\ast},z_{2}^{\ast},\cdots,z_{2t+3}^{\ast}\right)}\\
&~~~\times\left(\begin{bmatrix}u_{i}(0)\\u_{i}(1)\\\vdots\\u_{i}(t)\end{bmatrix}
-\begin{bmatrix}u_{j}(0)\\u_{j}(1)\\\vdots\\u_{j}(t)\end{bmatrix}\right)\\
&~~~+\left.\left[\frac{\Ds\partial \overline{g}^{t}}{\Ds\partial z_{t+3}},\frac{\Ds\partial \overline{g}^{t}}{\Ds\partial z_{t+4}},\cdots,\frac{\Ds\partial \overline{g}^{t}}{\Ds\partial z_{2t+3}}\right]\right|_{\left(z_{1},z_{2},\cdots,z_{2t+3}\right)=\left(z_{1}^{\ast},z_{2}^{\ast},\cdots,z_{2t+3}^{\ast}\right)}\\
&~~~\times\left(\begin{bmatrix}w_{i}(0)\\w_{i}(1)\\\vdots\\w_{i}(t)\end{bmatrix}
-\begin{bmatrix}w_{j}(0)\\w_{j}(1)\\\vdots\\w_{j}(t)\end{bmatrix}\right)\\
&~~~+\left.\frac{\Ds\partial \overline{g}^{t}}{\Ds\partial z_{1}}\right|_{\left(z_{1},z_{2},\cdots,z_{2t+3}\right)=\left(z_{1}^{\ast},z_{2}^{\ast},\cdots,z_{2t+3}^{\ast}\right)}(\delta_j-\delta_j)
\endaligned
\end{equation}

\noindent where
\[\aligned
\left(z_{1}^{\ast},z_{2}^{\ast},\cdots,z_{2t+3}^{\ast}\right)
&=\overline{\varpi}\left(y_i(0),u_{i}(0),\cdots,u_{i}(t),w_{i}(0),\cdots,w_{i}(t)\right)\\
&~~~+\left(1-\overline{\varpi}\right)\left(y_j(0),u_{j}(0),\cdots,u_{j}(t),w_{j}(0),\cdots,w_{j}(t)\right)
\endaligned\]

\noindent for some $\overline{\varpi}\in[0,1]$. Clearly, (\ref{eq0300}) can be rewritten in a compact form of (\ref{eq036}). Moreover, by setting $\beta_{\theta}=\max_{t\in\Z_{T-1}}\beta_{\theta}(t)$, the boundedness results of (\ref{eq8}), (\ref{eq04}) and (\ref{eq037}) can also be obtained.

\end{appendices}


\vspace{-1.8cm}

%
%
%
%
\end{document}